\documentclass[11pt,letterpaper]{article}

\usepackage{amsmath,amssymb,amsthm,graphicx,psfrag,times}

\newtheorem{theorem}{Theorem}
\newtheorem{lemma}{Lemma}
\newtheorem{proposition}{Proposition}

\newcommand{\etal}{et~al.}
\newcommand{\opt}{\mathrm{opt}}
\newcommand{\val}{\mathrm{val}}
\newcommand{\la}{\mathrm{la}}

\newcommand{\mis}[1]{\textup{Max-IS-$#1$}}
\newcommand{\mvc}[1]{\textup{Min-VC-$#1$}}
\newcommand{\sat}[2]{\textup{E$#1$-Occ-Max-E$#2$-SAT}}
\newcommand{\msr}[1]{\textup{MSR-$#1$}}
\newcommand{\cmsr}[1]{\textup{CMSR-$#1$}}
\newcommand{\gapmsr}[2]{\textup{$#1$-gap-MSR-$#2$}}
\newcommand{\gapcmsr}[2]{\textup{$#1$-gap-CMSR-$#2$}}
\newcommand{\dm}[1]{\textup{$#1$-Dimensional-Matching}}

\newcommand{\C}{{\mathcal{C}}}

\newcommand{\tl}[1]{\stackrel {#1}\blacktriangleleft}
\newcommand{\tr}[1]{\stackrel {#1}\blacktriangleright}
\newcommand{\fl}[1]{\stackrel {#1}\vartriangleleft}
\newcommand{\fr}[1]{\stackrel {#1}\vartriangleright}
\newcommand{\vl}[1]{\stackrel {#1}<}
\newcommand{\vr}[1]{\stackrel {#1}>}
\newcommand{\xl}[1]{\stackrel {#1}\sqsubset}
\newcommand{\xr}[1]{\stackrel {#1}\sqsupset}
\newcommand{\yl}[1]{\stackrel {#1}\Subset}
\newcommand{\yr}[1]{\stackrel {#1}\Supset}
\newcommand{\zl}[1]{\stackrel {#1}\subset}
\newcommand{\zr}[1]{\stackrel {#1}\supset}
\newcommand{\xs}{\;\;}

\begin{document}

\title{\bf Inapproximability of maximal strip recovery\thanks{%
This research was supported in part by NSF grant DBI-0743670.
A preliminary version of this paper appeared in two parts~\cite{Ji09,Ji10}
in the Proceedings of the 20th International Symposium on Algorithms and Computation (ISAAC 2009)
and the Proceedings of the 4th International Frontiers of Algorithmics Workshop (FAW 2010).}}

\author{Minghui Jiang
\medskip\\
Department of Computer Science, Utah State University, Logan, UT 84322, USA\\
\texttt{mjiang@cc.usu.edu}}

\maketitle

\begin{abstract}
In comparative genomic, the first step of sequence analysis
is usually to decompose two or more genomes into syntenic blocks
that are segments of homologous chromosomes.
For the reliable recovery of syntenic blocks,
noise and ambiguities in the genomic maps need to be removed first.
Maximal Strip Recovery (MSR) is an optimization problem
proposed by Zheng, Zhu, and Sankoff
for reliably recovering syntenic blocks from genomic maps
in the midst of noise and ambiguities.
Given $d$ genomic maps as sequences of gene markers,
the objective of \msr{d} is to find $d$ subsequences,
one subsequence of each genomic map,
such that the total length of syntenic blocks
in these subsequences is maximized.
For any constant $d \ge 2$,
a polynomial-time $2d$-approximation for \msr{d} was previously known.
In this paper, we show that for any $d \ge 2$,
\msr{d} is APX-hard,
even for the most basic version of the problem
in which all gene markers are distinct and appear in positive orientation
in each genomic map.
Moreover, we provide the first explicit lower bounds
on approximating \msr{d} for all $d \ge 2$.
In particular,
we show that \msr{d} is NP-hard to approximate within $\Omega(d/\log d)$.
From the other direction,
we show that the previous $2d$-approximation for \msr{d}
can be optimized into a polynomial-time algorithm
even if $d$ is not a constant but is part of the input.
We then extend our inapproximability results to several related problems
including \cmsr{d}, \gapmsr{\delta}{d}, and \gapcmsr{\delta}{d}.
\end{abstract}



\section{Introduction}

In comparative genomic, the first step of sequence analysis
is usually to decompose two or more genomes into syntenic blocks
that are segments of homologous chromosomes.
For the reliable recovery of syntenic blocks,
noise and ambiguities in the genomic maps need to be removed first.
A genomic map is a sequence of gene markers.
A gene marker appears in a genomic map
in either positive or negative orientation.
Given $d$ genomic maps,
\emph{Maximal Strip Recovery} (\msr{d}) is the problem of
finding $d$ subsequences, one subsequence of each genomic map,
such that the total length of strips of these subsequences
is maximized~\cite{ZZS07,CFJZ09}.
Here a \emph{strip} is a maximal string of at least two markers
such that either the string itself or its signed reversal appears
contiguously as a substring in each of the $d$ subsequences in the solution.
Without loss of generality,
we can assume that all markers appear in positive orientation
in the first genomic map.

For example, the two genomic maps
(the markers in negative orientation are underlined)
\begin{gather*}
1 \quad 2 \quad 3 \quad 4 \quad 5 \quad 6 \quad 7 \quad 8 \quad 9 \quad 10 \quad 11 \quad 12
\\
\underline{8} \quad \underline{5} \quad \underline{7} \quad \underline{6} \quad 4 \quad 1 \quad 3 \quad 2 \quad \underline{12} \quad \underline{11} \quad \underline{10} \quad 9
\end{gather*}
have two subsequences
\begin{gather*}
1 \quad 3
\quad\quad\quad
6 \quad 7 \quad 8
\quad\quad\quad
10 \quad 11 \quad 12
\\
\underline{8} \quad \underline{7} \quad \underline{6}
\quad\quad\quad
1 \quad 3
\quad\quad\quad
\underline{12} \quad \underline{11} \quad \underline{10}
\end{gather*}
of the maximum total strip length $8$.
The strip $\langle 1, 3 \rangle$ is positive and forward in both subsequences;
the other two strips $\langle 6, 7, 8 \rangle$ and $\langle 10, 11, 12 \rangle$
are positive and forward in the first subsequence,
but are negative and backward in the second subsequence.
Intuitively,
the strips are syntenic blocks,
and the deleted markers not in the strips
are noise and ambiguities in the genomic maps.

The problem \msr{2} was introduced by Zheng, Zhu, and Sankoff~\cite{ZZS07},
and was later generalized to \msr{d} for any $d \ge 2$
by Chen, Fu, Jiang, and Zhu~\cite{CFJZ09}.
For \msr{2},
Zheng \etal~\cite{ZZS07} presented
a potentially exponential-time heuristic
that solves a subproblem of Maximum-Weight Clique.
For \msr{d},
Chen \etal~\cite{CFJZ09} presented a $2d$-approximation
based on Bar-Yehuda \etal's fractional local-ratio algorithm
for Maximum-Weight Independent Set in $d$-interval graphs~\cite{BHN+06};
the running time of this $2d$-approximation algorithm is polynomial
if $d$ is a constant.

On the complexity side,
Chen \etal~\cite{CFJZ09} showed that several close variants of
the problem \msr{d} are intractable.
In particular, they showed that (i) \msr{2} is NP-complete
if duplicate markers are allowed in each genomic map,
and that (ii) \msr{3} is NP-complete
even if the markers in each genomic map are distinct.
The complexity of \msr{2} with no duplicates, however,
was left as an open problem.

In the biological context,
a genomic map may contain duplicate markers as
a paralogy set~\cite[p.~516]{ZZS07}, but such maps are relatively rare.
Thus \msr{2} without duplicates is the most useful version of \msr{d}
in practice.
Theoretically,
\msr{2} without duplicates is the most basic and hence
the most interesting version of \msr{d}.
Also, the previous NP-hardness proofs of
both (i) \msr{2} with duplicates and (ii) \msr{3} without duplicates~\cite{CFJZ09}
rely on the fact that a marker may appear in a genomic map
in either positive or negative orientation.
A natural question is
whether there is any version of \msr{d} that remains NP-hard
even if all markers in the genomic maps are in positive orientation.

We give a precise formulation of \emph{the most basic version}
of the problem \msr{d} as follows:
\begin{description}
\item{INSTANCE:}
Given $d$ sequences $G_i$, $1 \le i \le d$,
where each sequence is a permutation of $\langle 1,\ldots,n\rangle$.
\item{QUESTION:}
Find a subsequence $G'_i$ of each sequence $G_i$, $1 \le i \le d$,
and find a set of strips $S_j$,
where each strip is a sequence
of length at least two over the alphabet $\{1,\ldots n\}$,
such that each subsequence $G'_i$ is the concatenation of
the strips $S_j$ in some order,
and the total length of the strips $S_j$ is maximized.
\end{description}

The main result of this paper is the following theorem
that settles the computational complexity of the most basic version of
Maximal Strip Recovery,
and moreover provides the first explicit lower bounds
on approximating \msr{d} for all $d \ge 2$:

\begin{theorem}\label{thm:msrd}
\msr{d} for any $d \ge 2$ is APX-hard.
Moreover,
\msr{2}, \msr{3}, \msr{4}, and \msr{d}
are NP-hard to approximate
within
$1.000431$,
$1.002114$,
$1.010661$,
and
$\Omega(d/\log d)$,
respectively,
even if
all markers are distinct and appear in positive orientation in each genomic map.
\end{theorem}

Recall that for any constant $d \ge 2$,
\msr{d} admits a polynomial-time $2d$-approximation algorithm~\cite{CFJZ09}.
Thus \msr{d} for any constant $d \ge 2$ is APX-complete.
Our following theorem gives a polynomial-time $2d$-approximation algorithm
for \msr{d} even if the number $d$ of genomic maps is not a constant
but is part of the input:

\begin{theorem}\label{thm:2d}
For any $d \ge 2$, there is a polynomial-time $2d$-approximation algorithm
for \msr{d} if all markers are distinct in each genomic map.
This holds even if $d$ is not a constant but is part of the input.
\end{theorem}

Compare the upper bound of $2d$ in Theorem~\ref{thm:2d}
and the asymptotic lower bound of $\Omega(d/\log d)$ in Theorem~\ref{thm:msrd}.

Maximal Strip Recovery~\cite{ZZS07,CFJZ09} is a maximization problem.
Wang and Zhu~\cite{WZ09}
introduced Complement Maximal Strip Recovery as a minimization problem.
Given $d$ genomic maps as input,
the problem \cmsr{d} is the same as the problem \msr{d} except that
the objective is minimizing the number of deleted markers not in the strips,
instead of maximizing the number of markers in the strips.
A natural question is whether
a polynomial-time approximation scheme may be obtained for this problem.
Our following theorem shows that unless NP~$=$~P,
\cmsr{d} cannot be approximated arbitrarily well:

\begin{theorem}\label{thm:cmsrd}
\cmsr{d} for any $d \ge 2$ is APX-hard.
Moreover,
\cmsr{2}, \cmsr{3}, \cmsr{4}, and \cmsr{d} for any $d \ge 173$
are NP-hard to approximate within
$1.000625$,
$1.0101215$,
$1.0202429$,
and
$\frac{7}{6}- O(\log d /d)$,
respectively,
even if
all markers are distinct and appear in positive orientation in each genomic map.
If the number $d$ of genomic maps is not a constant but is part of the input,
then \cmsr{d} is NP-hard to approximate within any constant less than
$10\sqrt5 - 21 = 1.3606\ldots$,
even if
all markers are distinct and appear in positive orientation in each genomic map.
\end{theorem}

Note the similarity between Theorem~\ref{thm:msrd} and Theorem~\ref{thm:cmsrd}.
In fact, our proof of Theorem~\ref{thm:cmsrd} uses
exactly the same constructions as our proof of Theorem~\ref{thm:msrd}.
The only difference is in the analysis of the approximation lower bounds.

Bulteau, Fertin, and Rusu~\cite{BFR09} recently proposed a restricted variant
of Maximal Strip Recovery called $\delta$-gap-MSR, which is \msr{2}
with the additional constraint that at most $\delta$ markers
may be deleted between any two adjacent markers of a strip in each genomic map.
We now define \gapmsr{\delta}{d} and \gapcmsr{\delta}{d} as
the restricted variants of the two problems \msr{d} and \cmsr{d}, respectively,
with the additional $\delta$-gap constraint.
Bulteau \etal~\cite{BFR09} proved that
\gapmsr{\delta}{2} is APX-hard for any $\delta \ge 2$,
and is NP-hard for $\delta = 1$.
We extend our proofs of Theorem~\ref{thm:msrd} and Theorem~\ref{thm:cmsrd}
to obtain the following theorem on \gapmsr{\delta}{d} and \gapcmsr{\delta}{d}
for any $\delta \ge 2$:

\begin{theorem}\label{thm:gap}
Let $\delta \ge 2$. Then
\begin{enumerate}

\item[\textup{(1)}]
\gapmsr{\delta}{d} for any $d \ge 2$ is APX-hard.
Moreover,
\gapmsr{\delta}{2}, \gapmsr{\delta}{3}, \gapmsr{\delta}{4},
and \gapmsr{\delta}{d}
are NP-hard to approximate
within
$1.000431$,
$1.002114$,
$1.010661$,
and
$d/2^{O(\sqrt{\log d})}$,
respectively,
even if
all markers are distinct and appear in positive orientation in each genomic map.

\item[\textup{(2)}]
\gapcmsr{\delta}{d} for any $d \ge 2$ is APX-hard.
Moreover,
\gapcmsr{\delta}{2}, \gapcmsr{\delta}{3}, \gapcmsr{\delta}{4},
and \gapcmsr{\delta}{d} for any $d \ge 173$
are NP-hard to approximate within
$1.000625$,
$1.0101215$,
$1.0202429$,
and
$\frac{7}{6}- O(\log d /d)$,
respectively,
even if
all markers are distinct and appear in positive orientation in each genomic map.
If the number $d$ of genomic maps is not a constant but is part of the input,
then \gapcmsr{\delta}{d} is NP-hard to approximate within any constant less than
$10\sqrt5 - 21 = 1.3606\ldots$,
even if
all markers are distinct and appear in positive orientation in each genomic map.

\end{enumerate}
\end{theorem}

We refer to~\cite{CZZS07,Ji10c,BFJR10} for some related results.
Maximal Strip Recovery is a typical combinatorial problem
in biological sequence analysis, in particular, genome rearrangement.
The earliest inapproximability result for genome rearrangement problems
is due to Berman and Karpinski~\cite{BK99},
who proved that Sorting by Reversals
is NP-hard to approximate within any constant less than $\frac{1237}{1236}$.
More recently,
Zhu and Wang~\cite{ZW06} proved that
Translocation Distance
is NP-hard to approximate within any constant less than $\frac{5717}{5716}$.
Similar inapproximability results have also been obtained
for other important problems in bioinformatics.
For example,
Nagashima and Yamazaki~\cite{NY04} proved that Non-overlapping Local Alignment
is NP-hard to approximate within any constant less than $\frac{8668}{8665}$,
and Manthey~\cite{Ma05} proved that Multiple Sequence Alignment
with weighted sum-of-pairs score is APX-hard for arbitrary metric
scoring functions over the binary alphabet.


The rest of this paper is organized as follows.
We first review some preliminaries in Section~\ref{sec:pre}.
Then,
in Sections \ref{sec:msr4}, \ref{sec:msr3}, \ref{sec:msr2}, and \ref{sec:msrd},
we show that \msr{d} for any $d \ge 2$ is APX-hard,
and prove explicit approximation lower bounds.
(For any two constants $d$ and $d'$ such that $d' > d \ge 2$,
the problem \msr{d} is a special case of the problem \msr{d'}
with $d'-d$ redundant genomic maps.
Thus the APX-hardness of \msr{2} implies
the APX-hardness of \msr{d} for all constants $d \ge 2$.
To present the ideas progressively,
however,
we show that \msr{4}, \msr{3}, and \msr{2} are APX-hard
by three different L-reductions of increasing sophistication.)
In Section~\ref{sec:2d}, we present a $2d$-approximation algorithm for \msr{d}
that runs in polynomial time even if the number $d$ of genomic maps
is not a constant but is part of the input.
In Section~\ref{sec:more}, we present inapproximability results for \cmsr{d},
\gapmsr{\delta}{d}, and \gapcmsr{\delta}{d}.
We conclude with remarks in Section~\ref{sec:remarks}.

\section{Preliminaries}
\label{sec:pre}

\paragraph{L-reduction.}

Given two optimization problems X and Y,
an \emph{L-reduction}~\cite{PY91} from X to Y consists of
two polynomial-time functions $f$ and $g$
and two positive constants $\alpha$ and $\beta$
satisfying the following two properties:

\begin{enumerate}
\item
For every instance $x$ of X,
$f(x)$ is an instance of Y such that
\begin{equation}\label{eq:L1}
\opt(f(x)) \le \alpha \cdot \opt(x),
\end{equation}
\item
For every feasible solution $y$ to $f(x)$,
$g(y)$ is a feasible solution to $x$ such that
\begin{equation}\label{eq:L2}
|\opt(x) - \val(g(y))| \le \beta \cdot |\opt(f(x)) - \val(y)|.
\end{equation}
\end{enumerate}

Here
$\opt(x)$ denotes the value of the optimal solution to an instance $x$,
and
$\val(y)$ denotes the value of a solution $y$.
The two properties of L-reduction imply
the following inequality on the relative errors of approximation:
$$
\frac{|\opt(x) - \val(g(y))|}{\opt(x)}
\le \alpha\beta \cdot \frac{|\opt(f(x)) - \val(y)|}{\opt(f(x))}.
$$

A relative error of $\epsilon$
corresponds to an approximation factor of $1+\epsilon$
for a minimization problem,
and
corresponds to an approximation factor of $\frac1{1-\epsilon}$
for a maximization problem.
Thus we have the following propositions:
\begin{enumerate}

\item
For a minimization problem X and a minimization problem Y,
if X is NP-hard to approximate within
$1+\alpha\beta\epsilon$,
then Y is NP-hard to approximate within
$1+\epsilon$.

\item
For a maximization problem X and a maximization problem Y,
if X is NP-hard to approximate within
$\frac1{1-\alpha\beta\epsilon}$,
then Y is NP-hard to approximate within
$\frac1{1-\epsilon}$.

\item
For a minimization problem X and a maximization problem Y,
if X is NP-hard to approximate within
$1+\alpha\beta\epsilon$,
then Y is NP-hard to approximate within
$\frac1{1-\epsilon}$.

\item
For a maximization problem X and a minimization problem Y,
if X is NP-hard to approximate within
$\frac1{1-\alpha\beta\epsilon}$,
then Y is NP-hard to approximate within
$1+\epsilon$.

\end{enumerate}

\paragraph{APX-hard optimization problems.}

We review the complexities of some APX-hard optimization problems that will
be used in our reductions.
\begin{itemize}

\item
\mis{\Delta} is the problem Maximum Independent Set
in graphs of maximum degree $\Delta$.
\mis{3} is APX-hard; see~\cite{AK00}.
Moreover,
Chleb\'ik and Chleb\'ikov\'a~\cite{CC06} showed that
\mis{3} and \mis{4} are NP-hard to approximate within
$1.010661$
and
$1.0215517$,
respectively.
Trevisan~\cite{Tr01} showed that
\mis{\Delta} is NP-hard to approximate within
$\Delta/2^{O(\sqrt{\log \Delta})}$.

\item
\mvc{\Delta} is the problem Minimum Vertex Cover
in graphs of maximum degree $\Delta$.
\mvc{3} is APX-hard; see~\cite{AK00}.
Moreover,
Chleb\'ik and Chleb\'ikov\'a~\cite{CC06} showed that
\mvc{3} and \mvc{4}
are NP-hard to approximate within
$1.0101215$
and
$1.0202429$,
respectively,
and, for any $\Delta \ge 228$,
\mvc{\Delta} is NP-hard to approximate within
$\frac{7}{6} - O(\log \Delta /\Delta)$.
Dinur and Safra~\cite{DS05} showed that Minimum Vertex Cover
is NP-hard to approximate within any constant less than
$10\sqrt5 - 21 = 1.3606\ldots$.

\item
Given a set $X$ of $n$ variables and a set $\C$ of $m$ clauses,
where each variable has exactly $p$ literals
(in $p$ different clauses)
and each clause is the disjunction of exactly $q$ literals
(of $q$ different variables),
\sat{p}{q} is the problem of finding an assignment of $X$ that satisfies
the maximum number of clauses in $\C$.
Note that $np = mq$.
Berman and Karpinski~\cite{BK03} showed that
\sat{3}{2} is NP-hard to approximate within any constant less than
$\frac{464}{463}$.

\item
Given $d$ disjoint sets $V_i$ of vertices, $1 \le i \le d$,
and given a set $E \subseteq V_1 \times \cdots \times V_d$ of hyper-edges,
\dm{d} is the problem of finding
a maximum-cardinality subset $M \subseteq E$ of pairwise-disjoint hyper-edges.
Hazan, Safra, and Schwartz~\cite{HSS06} showed that
\dm{d} is NP-hard to approximate within $\Omega(d/\log d)$.

\end{itemize}

\paragraph{Linear forest and linear arboricity.}

A \emph{linear forest} is a graph in which every connected component is a path.
The \emph{linear arboricity} of a graph is
the minimum number of linear forests into which the edges of the graph
can be decomposed.
Akiyama, Exoo, and Harary~\cite{AEH80,AEH81} conjectured
that the linear arboricity of every graph $G$ of maximum degree $\Delta$
satisfies $\la(G) \le \lceil (\Delta + 1)/2 \rceil$.
This conjecture has been confirmed for graphs of small constant degrees,
and has been shown to be asymptotically correct
as $\Delta \to \infty$~\cite{Al88}.
In particular,
the proof of the conjecture for $\Delta = 3$ and $4$
are constructive~\cite{AEH80,AC81,AEH81}
and lead to polynomial-time algorithms for
decomposing any graph of maximum degree $\Delta = 3$ and $4$
into at most $\lceil (\Delta + 1)/2 \rceil = 2$ and $3$ linear forests,
respectively.
Also, the proof of the first upper bound on linear arboricity by
Akiyama, Exoo, and Harary~\cite{AEH81}
implies a simple polynomial-time algorithm for
decomposing any graph of maximum degree $\Delta$
into at most $\lceil 3\lceil \Delta/2 \rceil/2 \rceil$ linear forests.

Define
$$
f(\Delta) = \max_G f(G),
$$
where $G$ ranges over all graphs of maximum degree $\Delta$,
and $f(G)$ denotes the number of linear forests that
Akiyama, Exoo, and Harary's algorithm~\cite{AEH81}
decomposes $G$ into.
Then
\begin{equation}\label{eq:f}
\lceil (\Delta + 1)/2 \rceil
\le f(\Delta) \le
\lceil 3\lceil \Delta/2 \rceil/2 \rceil.
\end{equation}

\section{\msr{4} is APX-hard}
\label{sec:msr4}

In this section,
we prove that \msr{4} is APX-hard by a simple L-reduction from \mis{3}.
Before we present the L-reduction,
we first show that \msr{4} is NP-hard by a reduction
in the classical style, which is perhaps more familiar to most readers.
Throughout this paper, we follow this progressive format of presentation.

\subsection{NP-hardness reduction from \mis{3} to \msr{4}}

Let $G$ be a graph of maximum degree $3$.
Let $n$ be the number of vertices in $G$.
Partition the edges of $G$ into two linear forests $E_1$ and $E_2$.
Let $V_1$ and $V_2$ be the vertices of $G$
that are \emph{not} incident to any edges in $E_1$ and in $E_2$, respectively.
We construct four genomic maps
$G_{\rightarrow}$, $G_{\leftarrow}$, $G_1$, and $G_2$,
where each map is a permutation of the following $2n$ distinct markers
all in positive orientation:
\begin{itemize}

\item
$n$ pairs of vertex markers $\zl i$ and $\zr i$, $1 \le i \le n$.

\end{itemize}
$G_{\rightarrow}$ and $G_{\leftarrow}$
are concatenations of the $n$ pairs of vertex markers
with ascending and descending indices, respectively:
$$
\begin{array}{lc}

G_{\rightarrow}: & \zl 1 \; \zr 1 \quad \cdots \quad \zl n \; \zr n \\

G_{\leftarrow}: & \zl n \; \zr n \quad \cdots \quad \zl 1 \; \zr 1

\end{array}
$$
$G_1$ and $G_2$ are represented schematically as follows:
$$
\begin{array}{lcc}

G_1: & \langle E_1 \rangle & \langle V_1 \rangle \\

G_2: & \langle E_2 \rangle & \langle V_2 \rangle

\end{array}
$$
\begin{description}

\item
$\langle E_1 \rangle$ and $\langle E_2 \rangle$
consist of vertex markers of the vertices
incident to the edges in $E_1$ and $E_2$, respectively.
The markers of the vertices in each path $v_1 v_2 \ldots v_k$
are grouped together in an interleaving pattern:
for $1 \le i \le k$,
the left marker of $v_i$,
the right marker of $v_{i-1}$ (if $i > 1$),
the left marker of $v_{i+1}$ (if $i < k$),
and the right marker of $v_i$
are consecutive.

\item
$\langle V_1 \rangle$ and $\langle V_2 \rangle$
consist of vertex markers of the vertices in $V_1$ and $V_2$,
respectively.
The left marker and the right marker of each pair
are consecutive.

\end{description}
This completes the construction.
We refer to Figure~\ref{fig:msr4} (a) and (b) for an example.

\begin{figure}[htb]
\centering
\includegraphics{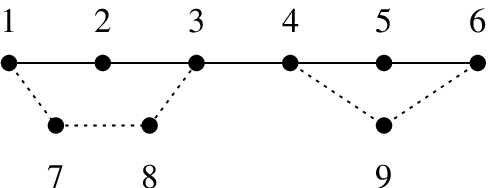}\\(a)

$$
\begin{array}{c}

\zl 1 \; \zr 1 \;
\zl 2 \; \zr 2 \;
\zl 3 \; \zr 3 \;
\zl 4 \; \zr 4 \;
\zl 5 \; \zr 5 \;
\zl 6 \; \zr 6 \;
\zl 7 \; \zr 7 \;
\zl 8 \; \zr 8 \;
\zl 9 \; \zr 9 \\

\zl 9 \; \zr 9 \;
\zl 8 \; \zr 8 \;
\zl 7 \; \zr 7 \;
\zl 6 \; \zr 6 \;
\zl 5 \; \zr 5 \;
\zl 4 \; \zr 4 \;
\zl 3 \; \zr 3 \;
\zl 2 \; \zr 2 \;
\zl 1 \; \zr 1 \\

\zl 1 \;
\zl 2 \;
\zr 1 \;
\zl 3 \;
\zr 2 \;
\zl 4 \;
\zr 3 \;
\zl 5 \;
\zr 4 \;
\zl 6 \;
\zr 5 \;
\zr 6 \quad

\zl 7 \; \zr 7 \;
\zl 8 \; \zr 8 \;
\zl 9 \; \zr 9 \\

\zl 1 \;
\zl 7 \;
\zr 1 \;
\zl 8 \;
\zr 7 \;
\zl 3 \;
\zr 8 \;
\zr 3 \quad

\zl 4 \;
\zl 9 \;
\zr 4 \;
\zl 6 \;
\zr 9 \;
\zr 6 \quad

\zl 2 \; \zr 2 \;
\zl 5 \; \zr 5

\end{array}
$$
(b)

$$
\begin{array}{c}

\zl 2 \; \zr 2 \;
\zl 4 \; \zr 4 \;
\zl 6 \; \zr 6 \;
\zl 8 \; \zr 8 \\

\zl 8 \; \zr 8 \;
\zl 6 \; \zr 6 \;
\zl 4 \; \zr 4 \;
\zl 2 \; \zr 2 \\

\zl 2 \; \zr 2 \;
\zl 4 \; \zr 4 \;
\zl 6 \; \zr 6 \;
\zl 8 \; \zr 8 \\

\zl 8 \; \zr 8 \;
\zl 4 \; \zr 4 \;
\zl 6 \; \zr 6 \;
\zl 2 \; \zr 2

\end{array}
$$
(c)

\caption{\small%
(a) The graph $G$:
$E_1$ is a single solid path $\langle 1,2,3,4,5,6 \rangle$,
$E_2$ consists of two dotted paths
$\langle 1,7,8,3 \rangle$ and $\langle 4,9,6 \rangle$,
$V_1 = \{ 7,8,9 \}$,
$V_2 = \{ 2,5 \}$.
(b) The four genomic maps $G_{\rightarrow}, G_{\leftarrow}, G_1, G_2$.
(c) The four subsequences of the genomic maps
corresponding to the independent set $\{ 2,4,6,8 \}$ in the graph.}
\label{fig:msr4}
\end{figure}

Two pairs of markers \emph{intersect} in a genomic map if
a marker of one pair appears between the two markers of the other pair.
The following property of our construction is obvious:

\begin{proposition}\label{prp:msr4}
Two vertices are adjacent in the graph $G$ if and only if
the corresponding two pairs of vertex markers intersect
in one of the two genomic maps $G_1, G_2$.
\end{proposition}

We say that four subsequences of the four genomic maps
$G_{\rightarrow}, G_{\leftarrow}, G_1, G_2$
are \emph{canonical} if
each strip of the subsequences is a pair of vertex markers.
We have the following lemma on canonical subsequences:

\begin{lemma}\label{lem:canon4}
In any four subsequences of the four genomic maps
$G_{\rightarrow}, G_{\leftarrow}, G_1, G_2$, respectively,
each strip must be a pair of vertex markers.
\end{lemma}

\begin{proof}
By construction,
a strip cannot include two vertex markers of different indices
because they appear in different orders
in $G_{\rightarrow}$ and in $G_{\leftarrow}$.
\end{proof}

The following lemma establishes the NP-hardness of \msr{4}:

\begin{lemma}\label{lem:iff4}
The graph $G$ has an independent set of at least $k$ vertices
if and only if
the four genomic maps
$G_{\rightarrow}, G_{\leftarrow}, G_1, G_2$
have four subsequences
whose total strip length $l$ is at least $2k$.
\end{lemma}

\begin{proof}
We first prove the ``only if'' direction.
Suppose that
the graph $G$ has an independent set of at least $k$ vertices.
We will show that
the four genomic maps
$G_{\rightarrow}, G_{\leftarrow}, G_1, G_2$
have four subsequences of total strip length at least $2k$.
By Proposition~\ref{prp:msr4},
the $k$ vertices in the independent set
correspond to $k$ pairs of vertex markers that do not intersect
each other in the genomic maps.
These $k$ pairs of vertex markers
induce a subsequence of length $2k$ in each genomic map.
In each subsequence,
the left marker and the right marker of each pair appear consecutively
and compose a strip.
Thus the total strip length is at least $2k$.
We refer to Figure~\ref{fig:msr4}(c) for an example.

We next prove the ``if'' direction.
Suppose that
the four genomic maps
$G_{\rightarrow}, G_{\leftarrow}, G_1, G_2$
have four subsequences of total strip length at least $2k$.
We will show that
the graph $G$ has an independent set of at least $k$ vertices.
By Lemma~\ref{lem:canon4},
each strip of the subsequences must be a pair of vertex markers.
Thus we obtain at least $k$ pairs of vertex markers
that do not intersect each other in the genomic maps.
Then, by Proposition~\ref{prp:msr4},
the corresponding set of at least $k$ vertices in the graph $G$
form an independent set.
\end{proof}

\subsection{L-reduction from \mis{3} to \msr{4}}

We present an L-reduction $(f, g, \alpha, \beta)$ from \mis{3} to \msr{4}
as follows.
The function $f$,
given a graph $G$ of maximum degree $3$,
constructs the four genomic maps
$G_{\rightarrow}, G_{\leftarrow}, G_1, G_2$
as in the NP-hardness reduction.
Let $k^*$ be the number of vertices in a maximum independent set in $G$,
and let $l^*$ be the maximum total strip length of any four subsequences of
$G_{\rightarrow}, G_{\leftarrow}, G_1, G_2$, respectively.
By Lemma~\ref{lem:iff4},
we have
$$
l^*= 2k^*.
$$
Choose $\alpha = 2$,
then property~\eqref{eq:L1} of L-reduction is satisfied.

The function $g$,
given four subsequences of the four genomic maps
$G_{\rightarrow}, G_{\leftarrow}, G_1, G_2$, respectively,
returns an independent set of vertices in the graph $G$
corresponding to
the pairs of vertex markers that are strips of the subsequences.
Let $l$ be the total strip length of the subsequences,
and let $k$ be the number of vertices in the independent set
returned by the function $g$.
Then $k \ge l/2$.
It follows that
$$
|k^* - k| =  k^* - k \le l^*/2 - l/2 = |l^* - l|/2.
$$
Choose $\beta = 1/2$,
then property~\eqref{eq:L2} of L-reduction is also satisfied.

We have obtained an L-reduction from \mis{3} to \msr{4}
with $\alpha\beta = 1$.
Chleb\'ik and Chleb\'ikov\'a~\cite{CC06} showed that
\mis{3} is NP-hard to approximate within $1.010661$.
It follows that \msr{4} is also NP-hard to approximate within $1.010661$.
The lower bound extends to \msr{d} for all constants $d \ge 4$.

The L-reduction from \mis{3} to \msr{4} can be obviously generalized:

\begin{lemma}\label{lem:mis}
Let $\Delta \ge 3$ and $d \ge 4$.
If there is a polynomial-time algorithm for
decomposing any graph of maximum degree $\Delta$ into $d-2$ linear forests,
then there is an L-reduction from \mis{\Delta} to \msr{d}
with constants $\alpha=2$ and $\beta=1/2$.
\end{lemma}

\section{\msr{3} is APX-hard}
\label{sec:msr3}

In this section,
we prove that \msr{3} is APX-hard by a slightly more sophisticated L-reduction
again from \mis{3}.

\subsection{NP-hardness reduction from \mis{3} to \msr{3}}

Let $G$ be a graph of maximum degree $3$.
Let $n$ be the number of vertices in $G$.
Partition the edges of $G$ into two linear forests $E_1$ and $E_2$.
Let $V_1$ and $V_2$ be the vertices of $G$
that are \emph{not} incident to any edges in $E_1$ and $E_2$, respectively.
We construct three genomic maps $G_0$, $G_1$, and $G_2$,
where each map is a permutation of the following $4n$ distinct markers
all in positive orientation:
\begin{itemize}

\item
$n$ pairs of vertex markers
$\zl i$ and $\zr i$, $1 \le i \le n$;

\item
$n$ pairs of dummy markers
$\xl i$ and $\xr i$, $1 \le i \le n$.

\end{itemize}
$G_0$ consists of the $2n$ pairs of vertex and dummy markers
in an alternating pattern:
$$
\zl 1 \; \zr 1 \quad \xl 1 \; \xr 1
\quad \cdots \quad
\zl n \; \zr n \quad \xl n \; \xr n
$$
$G_1$ and $G_2$ are represented schematically as follows:
$$
\begin{array}{lccc}

G_1: & \langle V_1 \rangle & \langle E_1 \rangle & \langle D \rangle \\

G_2: & \langle D \rangle & \langle E_2 \rangle & \langle V_2 \rangle

\end{array}
$$
\begin{description}

\item
$\langle E_1 \rangle$ and $\langle E_2 \rangle$
consist of vertex markers of the vertices
incident to the edges in $E_1$ and $E_2$, respectively.
The markers of the vertices in each path $v_1 v_2 \ldots v_k$
are grouped together in an interleaving pattern:
for $1 \le i \le k$,
the left marker of $v_i$,
the right marker of $v_{i-1}$ (if $i > 1$),
the left marker of $v_{i+1}$ (if $i < k$),
and the right marker of $v_i$
are consecutive.

\item
$\langle V_1 \rangle$ and $\langle V_2 \rangle$
consist of vertex markers of the vertices in $V_1$ and $V_2$,
respectively.
The left marker and the right marker of each pair
are consecutive.

\item
$\langle D \rangle$ is the reverse permutation of
the $n$ pairs of dummy markers:

$$
\xl n \; \xr n
\quad \cdots \quad
\xl 1 \; \xr 1
$$

\end{description}
This completes the construction.
We refer to Figure~\ref{fig:msr3} (a) and (b) for an example.

\begin{figure}[htb]
\centering
\includegraphics{forests.eps}\\(a)

$$
\begin{array}{c}

\zl 1 \; \zr 1 \; \xl 1 \; \xr 1 \;
\zl 2 \; \zr 2 \; \xl 2 \; \xr 2 \;
\zl 3 \; \zr 3 \; \xl 3 \; \xr 3 \;
\zl 4 \; \zr 4 \; \xl 4 \; \xr 4 \;
\zl 5 \; \zr 5 \; \xl 5 \; \xr 5 \;
\zl 6 \; \zr 6 \; \xl 6 \; \xr 6 \;
\zl 7 \; \zr 7 \; \xl 7 \; \xr 7 \;
\zl 8 \; \zr 8 \; \xl 8 \; \xr 8 \;
\zl 9 \; \zr 9 \; \xl 9 \; \xr 9 \\

\zl 7 \; \zr 7 \;
\zl 8 \; \zr 8 \;
\zl 9 \; \zr 9 \quad

\zl 1 \;
\zl 2 \;
\zr 1 \;
\zl 3 \;
\zr 2 \;
\zl 4 \;
\zr 3 \;
\zl 5 \;
\zr 4 \;
\zl 6 \;
\zr 5 \;
\zr 6 \quad

\xl 9 \; \xr 9 \;
\xl 8 \; \xr 8 \;
\xl 7 \; \xr 7 \;
\xl 6 \; \xr 6 \;
\xl 5 \; \xr 5 \;
\xl 4 \; \xr 4 \;
\xl 3 \; \xr 3 \;
\xl 2 \; \xr 2 \;
\xl 1 \; \xr 1 \\

\xl 9 \; \xr 9 \;
\xl 8 \; \xr 8 \;
\xl 7 \; \xr 7 \;
\xl 6 \; \xr 6 \;
\xl 5 \; \xr 5 \;
\xl 4 \; \xr 4 \;
\xl 3 \; \xr 3 \;
\xl 2 \; \xr 2 \;
\xl 1 \; \xr 1 \quad

\zl 1 \;
\zl 7 \;
\zr 1 \;
\zl 8 \;
\zr 7 \;
\zl 3 \;
\zr 8 \;
\zr 3 \quad

\zl 4 \;
\zl 9 \;
\zr 4 \;
\zl 6 \;
\zr 9 \;
\zr 6 \quad

\zl 2 \; \zr 2 \;
\zl 5 \; \zr 5

\end{array}
$$
(b)

$$
\begin{array}{c}

\xl 1 \; \xr 1 \;
\zl 2 \; \zr 2 \; \xl 2 \; \xr 2 \;
\xl 3 \; \xr 3 \;
\zl 4 \; \zr 4 \; \xl 4 \; \xr 4 \;
\xl 5 \; \xr 5 \;
\zl 6 \; \zr 6 \; \xl 6 \; \xr 6 \;
\xl 7 \; \xr 7 \;
\zl 8 \; \zr 8 \; \xl 8 \; \xr 8 \;
\xl 9 \; \xr 9 \\

\zl 8 \; \zr 8 \;
\zl 2 \; \zr 2 \;
\zl 4 \; \zr 4 \;
\zl 6 \; \zr 6 \quad

\xl 9 \; \xr 9 \;
\xl 8 \; \xr 8 \;
\xl 7 \; \xr 7 \;
\xl 6 \; \xr 6 \;
\xl 5 \; \xr 5 \;
\xl 4 \; \xr 4 \;
\xl 3 \; \xr 3 \;
\xl 2 \; \xr 2 \;
\xl 1 \; \xr 1 \\

\xl 9 \; \xr 9 \;
\xl 8 \; \xr 8 \;
\xl 7 \; \xr 7 \;
\xl 6 \; \xr 6 \;
\xl 5 \; \xr 5 \;
\xl 4 \; \xr 4 \;
\xl 3 \; \xr 3 \;
\xl 2 \; \xr 2 \;
\xl 1 \; \xr 1 \quad

\zl 8 \; \zr 8 \;
\zl 4 \; \zr 4 \;
\zl 6 \; \zr 6 \;
\zl 2 \; \zr 2 \;

\end{array}
$$
(c)

\caption{\small%
(a) The graph $G$:
$E_1$ is a single (solid) path $\langle 1,2,3,4,5,6 \rangle$,
$E_2$ consists of two (dotted) paths
$\langle 1,7,8,3 \rangle$ and $\langle 4,9,6 \rangle$,
$V_1 = \{ 7,8,9 \}$,
$V_2 = \{ 2,5 \}$.
(b) The three genomic maps $G_0, G_1, G_2$.
(c) The three subsequences of the genomic maps
corresponding to the independent set $\{ 2,4,6,8 \}$ in the graph.}
\label{fig:msr3}
\end{figure}

It is clear that Proposition~\ref{prp:msr4} still holds.
The following lemma on canonical subsequences
is analogous to Lemma~\ref{lem:canon4}:

\begin{lemma}\label{lem:canon3}
If the three genomic maps $G_0, G_1, G_2$ have three subsequences
of total strip length $l$,
then they must have three subsequences of total strip length at least $l$
such that
\textup{(i)}
each strip is either a pair of vertex markers or a pair of dummy markers,
and
\textup{(ii)}
each pair of dummy markers is a strip.
\end{lemma}

\begin{proof}
We present an algorithm that transforms the subsequences into canonical form
without reducing the total strip length.
By construction,
a strip cannot include both a dummy marker and a vertex marker
because they appear in different orders in $G_1$ and in $G_2$,
and a strip cannot include two dummy markers of different indices
because they appear in different orders in $G_0$ and in $G_1$ and $G_2$.
Suppose that a strip $S$ consists of vertex markers
of two or more different indices.
Then there must be two vertex markers $\mu$ and $\nu$
of different indices $i$ and $j$ that are consecutive in $S$.
Since the vertex markers and the dummy markers appear in $G_0$
in an alternating pattern with ascending indices,
we must have $i < j$.
Moreover, the pair of dummy markers of index $i$,
which appears between $\mu$ and $\nu$ in $G_0$,
must be missing from the subsequences.
Now cut the strip $S$ into $S_\mu$ and $S_\nu$ between $\mu$ and $\nu$.
If $S_\mu$ (resp.\ $S_\nu$) consists of only one marker $\mu$ (resp.\ $\nu$),
delete the lone marker from the subsequences
(recall that a strip must include at least two markers).
This decreases the total strip length by at most two.
Next insert the pair of dummy markers of index $i$
to the subsequences as a new strip.
This increases the total strip length by exactly two.
Repeat this operation
whenever a strip contains two vertex markers of different indices
and whenever a pair of dummy markers is missing from the subsequences,
then in $O(n)$ steps we obtain three subsequences of total strip length
at least $l$ in canonical form.
\end{proof}

The following lemma, analogous to Lemma~\ref{lem:iff4},
establishes the NP-hardness of \msr{3}:

\begin{lemma}\label{lem:iff3}
The graph $G$ has an independent set of at least $k$ vertices
if and only if
the three genomic maps $G_0, G_1, G_2$ have three subsequences
whose total strip length $l$ is at least $2(n+k)$.
\end{lemma}

\begin{proof}
We first prove the ``only if'' direction.
Suppose that
the graph $G$ has an independent set of at least $k$ vertices.
We will show that
the three genomic maps $G_0, G_1, G_2$
have three subsequences of total strip length at least $2(n+k)$.
By Proposition~\ref{prp:msr4},
the $k$ vertices in the independent set
correspond to $k$ pairs of vertex markers that do not intersect
each other in the genomic maps.
These $k$ pairs of vertex markers
together with the $n$ pairs of dummy markers
induce a subsequence of length $2(n+k)$ in each genomic map.
In each subsequence,
the left marker and the right marker of each pair appear consecutively
and compose a strip.
Thus the total strip length is at least $2(n+k)$.
We refer to Figure~\ref{fig:msr3}(c) for an example.

We next prove the ``if'' direction.
Suppose that
the three genomic maps $G_0, G_1, G_2$
have three subsequences of total strip length at least $2(n+k)$.
We will show that
the graph $G$ has an independent set of at least $k$ vertices.
By Lemma~\ref{lem:canon3},
the three genomic maps have three subsequences
of total strip length at least $2(n+k)$
such that each strip is a pair of markers.
Excluding the $n$ pairs of dummy markers,
we obtain at least $k$ pairs of vertex markers
that do not intersect each other in the genomic maps.
Then, by Proposition~\ref{prp:msr4},
the corresponding set of at least $k$ vertices in the graph $G$
form an independent set.
\end{proof}

\subsection{L-reduction from \mis{3} to \msr{3}}

We present an L-reduction $(f, g, \alpha, \beta)$ from \mis{3} to \msr{3}
as follows.
The function $f$,
given a graph $G$ of maximum degree $3$,
constructs the three genomic maps $G_0, G_1, G_2$
as in the NP-hardness reduction.
Let $k^*$ be the number of vertices in a maximum independent set in $G$,
and let $l^*$ be the maximum total strip length of any three subsequences of
$G_0, G_1, G_2$, respectively.
Since a simple greedy algorithm
(which repeatedly selects a vertex not adjacent to
the previously selected vertices)
finds an independent set of at least $n/(3+1)$ vertices
in the graph $G$ of maximum degree $3$,
we have $k^* \ge n/(3+1)$.
By Lemma~\ref{lem:iff3},
we have $l^*= 2(n+k^*)$.
It follows that
$$
l^*= 2(n+k^*) \le 2((3 + 1)k^*+k^*) = 2(3 + 2) k^* = 10 k^*.
$$
Choose $\alpha = 10$,
then property~\eqref{eq:L1} of L-reduction is satisfied.

The function $g$,
given three subsequences of the three genomic maps
$G_0, G_1, G_2$, respectively,
transforms the subsequences into canonical form
as in the proof of Lemma~\ref{lem:canon3},
then returns an independent set of vertices in the graph $G$
corresponding to
the pairs of vertex markers that are strips of the subsequences.
Let $l$ be the total strip length of the subsequences,
and let $k$ be the number of vertices in the independent set
returned by the function $g$.
Then $k \ge l/2 - n$.
It follows that
$$
|k^* - k| = k^* - k \le (l^*/2 - n) - (l/2 - n) = |l^* - l|/2.
$$
Choose $\beta = 1/2$,
then property~\eqref{eq:L2} of L-reduction is also satisfied.

We have obtained an L-reduction from \mis{3} to \msr{3}
with $\alpha\beta = 5$.
Chleb\'ik and Chleb\'ikov\'a~\cite{CC06} showed that
\mis{3} is NP-hard to approximate within
$1.010661 = \frac1{1 - (1 - 1/1.010661)}$.
It follows that \msr{3} is NP-hard to approximate within
$\frac1{1 - (1 - 1/1.010661)/5} = 1.002114\ldots$.

\section{\msr{2} is APX-hard}
\label{sec:msr2}

In this section,
we prove that \msr{2} is APX-hard by an L-reduction from \sat{p}{q}
with $p = 3$ and $q \ge 2$.

\subsection{NP-hardness reduction from \sat{p}{q} to \msr{2}}

Let $(X,\C)$ be an instance of \sat{p}{q},
where $X$ is a set of $n$ variables $x_i$, $1 \le i \le n$,
and $\C$ is a set of $m$ clauses $C_j$, $1 \le j \le m$.
Without loss of generality, assume that the $p$ literals of each variable
are neither all positive nor all negative.
Since $p = 3$, it follows that
each variable has either $2$ positive and $1$ negative literals,
or $1$ positive and $2$ negative literals.

We construct two genomic maps $G_1$ and $G_2$,
each map a permutation of $2(5n + m + qm + 2)$ distinct markers
all in positive orientation:
\begin{itemize}

\item
$1$ pair of variable markers $\vl{i}\;\vr{i}$
for each variable $x_i$, $1 \le i \le n$;

\item
$2$ pairs of true markers
$\tl{i,1}\;\tr{i,1}$ and $\tl{i,2}\;\tr{i,2}$
for each variable $x_i$, $1 \le i \le n$;

\item
$2$ pairs of false markers
$\fl{i,1}\;\fr{i,1}$ and $\fl{i,2}\;\fr{i,2}$
for each variable $x_i$, $1 \le i \le n$;

\item
$1$ pair of clause markers $\yl{j}\;\yr{j}$
for each clause $C_j$, $1 \le j \le m$;

\item
$q$ pairs of literal markers
$\zl{j,t}\;\zr{j,t}$, $1 \le t \le q$,
for each clause $C_j$, $1 \le j \le m$;

\item
$2$ pairs of dummy markers $\xl1\;\xr1$ and $\xl2\;\xr2$.

\end{itemize}
 
The construction is done in two steps:
first arrange the variable markers, the true/false markers,
the clause markers, and the dummy markers
into two sequences $\check G_1$ and $\check G_2$,
next insert the literal markers at appropriate positions in the two sequences
to obtain the two genomic maps $G_1$ and $G_2$.

The two sequences $\check G_1$ and $\check G_2$ are represented
schematically as follows:
$$
\begin{array}{ll}

\check G_1:
&
\quad \langle x_1 \rangle \quad \cdots \quad \langle x_n \rangle
\quad
\quad \xl 1 \quad \xr 1 \quad \xl 2 \quad \xr 2
\quad
\quad \yl 1 \quad \yr 1 \quad \cdots \quad \yl m \quad \yr m
\quad
\quad \vl 1 \quad \vr 1 \quad \cdots \quad \vl n \quad \vr n \\

\check G_2:
&
\quad \langle x_n \rangle \quad \cdots \quad \langle x_1 \rangle
\quad
\quad \yl m \quad \yr m \quad \cdots \quad \yl 1 \quad \yr 1
\quad
\quad \xl 2 \quad \xr 2 \quad \xl 1 \quad \xr 1

\end{array}
$$
For each variable $x_i$,
$\langle x_i \rangle$
consists of the corresponding four pairs of true/false markers
$\tl{i,1}\;\tr{i,1}$ $\tl{i,2}\;\tr{i,2}$
$\fl{i,1}\;\fr{i,1}$ $\fl{i,2}\;\fr{i,2}$
in $\check G_1$ and $\check G_2$,
and in addition the pair of variable markers $\vl{i}\;\vr{i}$ in $\check G_2$.
These markers are arranged in the two sequences in a special pattern as follows
(the indices $i$ are omitted for simpler notations):
$$
\begin{array}{c}

\fl{1} \quad \tl{2} \quad \fr{1} \quad \tr{2}
\quad\quad
\fl{2} \quad \tl{1} \quad \fr{2} \quad \tr{1} \\

\tl{1} \quad \fl{1} \quad \tr{1} \quad \fr{1}
\quad \vl{} \quad \vr{} \quad
\fl{2} \quad \tl{2} \quad \fr{2} \quad \tr{2}

\end{array}
$$

Now insert the literal markers to the two sequences
$\check G_1$ and $\check G_2$ to obtain the two genomic maps $G_1$ and $G_2$.
First, $\check G_1 \to G_1$.
For each positive literal (resp.\ negative literal) of a variable $x_i$
that occurs in a clause $C_j$,
place a pair of literal markers $\zl{j,t}\;\zr{j,t}$, $1 \le t \le q$,
around a false marker $\fl{i,s}$ (resp.\ true marker $\tr{i,s}$),
$1 \le s \le 2$.
The four possible positions of the three pairs of literal markers
of each variable $x_i$ are as follows:
$$
\begin{array}{c}

\zl{}\fl{1}\zr{} \quad \tl{2} \quad \fr{1} \quad \zl{}\tr{2}\zr{}
\quad\quad
\zl{}\fl{2}\zr{} \quad \tl{1} \quad \fr{2} \quad \zl{}\tr{1}\zr{} \\

\tl{1} \quad \fl{1} \quad \tr{1} \quad \fr{1}
\quad \vl{} \quad \vr{} \quad
\fl{2} \quad \tl{2} \quad \fr{2} \quad \tr{2}

\end{array}
$$

Next, $\check G_2 \to G_2$.
Without loss of generality, assume that
the $q$ pairs of literal markers of each clause $C_j$
appear in $G_1$ with ascending indices:
$$
\zl{j,1} \quad \zr{j,1} \quad \cdots \quad \zl{j,q} \quad \zr{j,q}
$$
Insert the $q$ pairs of literal markers in $G_2$
immediately after the pair of clause markers $\yl{j}\;\yr{j}$,
in an interleaving pattern:
$$
\zl{j,q} \quad \cdots \quad \zl{j,1} \quad \zr{j,q} \quad \cdots \quad \zr{j,1}
$$
This completes the construction.
We refer to Figure~\ref{fig:msr2} (a) and (b) for an example of the two steps.

\begin{figure}[htb]
\centering

\resizebox{\linewidth}{!}{\parbox[c]{\linewidth}{$$\begin{array}{c}

\fl{1,1}\xs\tl{1,2}\xs\fr{1,1}\xs\tr{1,2}\xs
\fl{1,2}\xs\tl{1,1}\xs\fr{1,2}\xs\tr{1,1}\xs
\fl{2,1}\xs\tl{2,2}\xs\fr{2,1}\xs\tr{2,2}\xs
\fl{2,2}\xs\tl{2,1}\xs\fr{2,2}\xs\tr{2,1}\xs
\xl1\xs\xr1\xs \xl2\xs\xr2\xs
\yl1\xs\yr1\xs \yl2\xs\yr2\xs \yl3\xs\yr3\xs
\vl1\xs\vr1\xs \vl2\xs\vr2 \\

\tl{2,1}\xs\fl{2,1}\xs\tr{2,1}\xs\fr{2,1}\xs
\vl{2}\xs\vr{2}\xs
\fl{2,2}\xs\tl{2,2}\xs\fr{2,2}\xs\tr{2,2}\xs
\tl{1,1}\xs\fl{1,1}\xs\tr{1,1}\xs\fr{1,1}\xs
\vl{1}\xs\vr{1}\xs
\fl{1,2}\xs\tl{1,2}\xs\fr{1,2}\xs\tr{1,2}\xs
\yl3\xs\yr3\xs \yl2\xs\yr2\xs \yl1\xs\yr1\xs
\xl2\xs\xr2\xs \xl1\xs\xr1

\end{array}$$}}\\(a)

\resizebox{\linewidth}{!}{\parbox[c]{1.43\linewidth}{$$\begin{array}{c}

\zl{1,1}\xs\fl{1,1}\xs\zr{1,1}\xs\tl{1,2}\xs\fr{1,1}\xs\zl{3,1}\xs\tr{1,2}\xs\zr{3,1}\xs
\zl{2,1}\xs\fl{1,2}\xs\zr{2,1}\xs\tl{1,1}\xs\fr{1,2}\xs\tr{1,1}\xs
\zl{1,2}\xs\fl{2,1}\xs\zr{1,2}\xs\tl{2,2}\xs\fr{2,1}\xs\zl{2,2}\xs\tr{2,2}\xs\zr{2,2}\xs
\fl{2,2}\xs\tl{2,1}\xs\fr{2,2}\xs\zl{3,2}\xs\tr{2,1}\xs\zr{3,2}\xs
\xl1\xs\xr1\xs \xl2\xs\xr2\xs
\yl1\xs\yr1\xs \yl2\xs\yr2\xs \yl3\xs\yr3\xs
\vl1\xs\vr1\xs \vl2\xs\vr2 \\

\tl{2,1}\xs\fl{2,1}\xs\tr{2,1}\xs\fr{2,1}\xs
\vl{2}\xs\vr{2}\xs
\fl{2,2}\xs\tl{2,2}\xs\fr{2,2}\xs\tr{2,2}\xs
\tl{1,1}\xs\fl{1,1}\xs\tr{1,1}\xs\fr{1,1}\xs
\vl{1}\xs\vr{1}\xs
\fl{1,2}\xs\tl{1,2}\xs\fr{1,2}\xs\tr{1,2}\xs
\yl3\xs\yr3\xs
\zl{3,2}\xs\zl{3,1}\xs\zr{3,2}\xs\zr{3,1}\xs
\yl2\xs\yr2\xs
\zl{2,2}\xs\zl{2,1}\xs\zr{2,2}\xs\zr{2,1}\xs
\yl1\xs\yr1\xs
\zl{1,2}\xs\zl{1,1}\xs\zr{1,2}\xs\zr{1,1}\xs
\xl2\xs\xr2\xs \xl1\xs\xr1

\end{array}$$}}\\(b)

\resizebox{\linewidth}{!}{\parbox[c]{\linewidth}{$$\begin{array}{c}

\zl{1,1}\xs\zr{1,1}\xs\tl{1,2}\xs\tr{1,2}\xs
\zl{2,1}\xs\zr{2,1}\xs\tl{1,1}\xs\tr{1,1}\xs
\fl{2,1}\xs\fr{2,1}\xs
\fl{2,2}\xs\fr{2,2}\xs\zl{3,2}\xs\zr{3,2}\xs
\xl1\xs\xr1\xs \xl2\xs\xr2\xs
\yl1\xs\yr1\xs \yl2\xs\yr2\xs \yl3\xs\yr3\xs
\vl1\xs\vr1\xs \vl2\xs\vr2 \\

\fl{2,1}\xs\fr{2,1}\xs
\vl{2}\xs\vr{2}\xs
\fl{2,2}\xs\fr{2,2}\xs
\tl{1,1}\xs\tr{1,1}\xs
\vl{1}\xs\vr{1}\xs
\tl{1,2}\xs\tr{1,2}\xs
\yl3\xs\yr3\xs
\zl{3,2}\xs\zr{3,2}\xs
\yl2\xs\yr2\xs
\zl{2,1}\xs\zr{2,1}\xs
\yl1\xs\yr1\xs
\zl{1,1}\xs\zr{1,1}\xs
\xl2\xs\xr2\xs \xl1\xs\xr1

\end{array}$$}}\\(c)

\resizebox{\linewidth}{!}{\parbox[c]{\linewidth}{$$\begin{array}{c}

\zl{1,1}\xs\zr{1,1}\xs\tl{1,2}\xs\tr{1,2}\xs
\tl{1,1}\xs\tr{1,1}\xs
\fl{2,1}\xs\fr{2,1}\xs\zl{2,2}\xs\zr{2,2}\xs
\fl{2,2}\xs\fr{2,2}\xs\zl{3,2}\xs\zr{3,2}\xs
\xl1\xs\xr1\xs \xl2\xs\xr2\xs
\yl1\xs\yr1\xs \yl2\xs\yr2\xs \yl3\xs\yr3\xs
\vl1\xs\vr1\xs \vl2\xs\vr2 \\

\fl{2,1}\xs\fr{2,1}\xs
\vl{2}\xs\vr{2}\xs
\fl{2,2}\xs\fr{2,2}\xs
\tl{1,1}\xs\tr{1,1}\xs
\vl{1}\xs\vr{1}\xs
\tl{1,2}\xs\tr{1,2}\xs
\yl3\xs\yr3\xs
\zl{3,2}\xs\zr{3,2}\xs
\yl2\xs\yr2\xs
\zl{2,2}\xs\zr{2,2}\xs
\yl1\xs\yr1\xs
\zl{1,1}\xs\zr{1,1}\xs
\xl2\xs\xr2\xs \xl1\xs\xr1

\end{array}$$}}\\(d)

\caption{\small%
\msr{2} construction for the \sat{3}{2} instance
$C_1 = x_1 \lor x_2$, $C_2 = x_1 \lor \bar x_2$,
and $C_3 = \bar x_1 \lor \bar x_2$.
(a)
The two sequences $\check G_1$ and $\check G_2$.
(b)
The two genomic maps $G_1$ and $G_2$.
(c) Two canonical subsequences
for the assignment $x_1 = \mathit{true}$ and $x_2 = \mathit{false}$.
(d) Two other canonical subsequences
for the assignment $x_1 = \mathit{true}$ and $x_2 = \mathit{false}$.
}
\label{fig:msr2}
\end{figure}

We say that
two subsequences of the two genomic maps $G_1$ and $G_2$
are \emph{canonical} if
each strip of the two subsequences is a pair of markers.
We refer to Figure~\ref{fig:msr2} (c) and (d)
for two examples of canonical subsequences.
The following lemma on canonical subsequences
is analogous to Lemma~\ref{lem:canon4} and Lemma~\ref{lem:canon3}:

\begin{lemma}\label{lem:canon2}
If the two genomic maps $G_1$ and $G_2$ have
two subsequences of total strip length $l$,
then they must have
two subsequences of total strip length at least $l$
such that each strip is a pair of markers and, moreover,
\textup{(i)}
the two pairs of dummy markers are two strips,
\textup{(ii)}
the $m$ pairs of clause markers and the $n$ pairs of variable markers are $m+n$ strips,
\textup{(iii)}
at most one pair of literal markers of each clause is a strip,
\textup{(iv)}
either both pairs of true markers or both pairs of false markers of each variable are two strips.
\end{lemma}

\begin{proof}
We present an algorithm that transforms the subsequences into canonical form
without reducing the total strip length.
The algorithm performs incremental operations on the subsequences
such that the following eight conditions are satisfied progressively:

\textbf{\em 1. Each strip that includes a dummy marker is a pair of dummy markers.}
A strip cannot include two dummy markers of different indices
because they appear in different orders in $G_1$ and in $G_2$.
Note that in $G_2$ the dummy markers appear after the other markers.
Suppose that a strip $S$ includes both a dummy marker and a non-dummy marker.
Then there must be a non-dummy marker $\mu$ and a dummy marker $\nu$
consecutive in $S$.
Since the two pairs of dummy markers
appear consecutively but in different orders in $G_1$ and in $G_2$,
one of the two pairs must appear between $\mu$ and $\nu$
either in $G_1$ or in $G_2$.
This pair is hence missing from the subsequences.
Now cut the strip $S$ into $S_\mu$ and $S_\nu$ between $\mu$ and $\nu$.
If $S_\mu$ (resp.\ $S_\nu$) consists of only one marker $\mu$ (resp.\ $\nu$),
delete the lone marker from the subsequences
(recall that a strip must include at least two markers).
This decreases the total strip length by at most two.
Next insert the missing pair of dummy markers
to the subsequences.
This pair of dummy markers becomes either a new strip by itself,
or part of a longer strip
(recall that a strip must be maximal).
In any case, the insertion increases the total strip length by exactly two.
Overall, this \emph{cut-delete-insert} operation
(also used in Lemma~\ref{lem:canon3})
does not reduce the total strip length.
After the first operation,
a second operation may be necessary.
But since each operation here deletes only lone markers
(in $S_\mu$ and $S_\nu$) and inserts always a pair of markers,
the pair inserted by one operation
is never deleted by a subsequent operation.
Thus at most two operations are sufficient to transform the subsequences
until each strip that includes a dummy marker is indeed a pair of dummy markers.

\textbf{\em 2. The two pairs of dummy markers are two strips.}
Suppose that the subsequences do not have both pairs of dummy markers as strips.
Then, by condition~1,
we must have either both pairs of dummy markers missing from the subsequences,
or one pair missing and the other pair forming a strip.
Note that in $G_1$ the dummy markers separate
the true/false and literal markers on the left
from the clause and variable markers on the right,
and that in $G_2$ the dummy markers appear after the other markers.
If the missing dummy markers do not disrupt any existing strips in $G_1$,
then simply insert each missing pair to the subsequences as a new strip.
Otherwise, there must be a true/false or literal marker $\mu$
and a clause or variable marker $\nu$ consecutive in a strip $S$,
such that both pairs of dummy markers appear in $G_1$ between $\mu$ and $\nu$
and hence are missing from the subsequences.
Cut the strip $S$ between $\mu$ and $\nu$,
delete any lone markers if necessary,
then insert the two pairs of dummy markers to the subsequences
as two new strips.

\textbf{\em 3. Each strip that includes a clause or variable marker is a pair of clause markers or a pair of variable markers.}
Note that in $G_1$ the clause and variable markers
are separated by the dummy markers from the other markers.
Thus, by condition~2,
a strip that includes a clause or variable marker
cannot include any markers of the other types.
Also, a strip cannot include two clause markers of different clauses,
or two variable markers of different variables,
or a clause marker and a variable marker,
because these combinations appear in different orders in $G_1$ and in $G_2$.
Thus this condition is automatically satisfied after conditions 1 and 2.

\textbf{\em 4. The $m$ pairs of clause markers and the $n$ pairs of variable markers are $m+n$ strips.}
Suppose that the subsequences do not have
all $m+n$ pairs of clause and variable markers as $m+n$ strips.
By condition~3,
the clause and variable markers in the subsequences must be in pairs,
each pair forming a strip.
Then the clause and variable markers missing from the subsequences
must be in pairs too.
For each missing pair of clause or variable markers,
if the pair does not disrupt any existing strips in $G_2$,
then simply insert it to the subsequences as a new strip.
Otherwise, there must be two true/false or literal markers $\mu$ and $\nu$
consecutive in a strip $S$,
such that the missing pair appears in $G_2$ between $\mu$ and $\nu$.
Cut the strip $S$ between $\mu$ and $\nu$,
delete any lone markers if necessary,
then insert each missing pair of clause markers
between $\mu$ and $\nu$
to the subsequences as a new strip.

\textbf{\em 5. Each strip that includes a literal marker is a pair of literal markers.}
Note that in $G_2$ the dummy and clause markers separate
the literals markers from the other markers,
and separate the literal markers of different clauses from each other.
Thus, by conditions 2 and 4,
a strip cannot include both a literal marker and a non-literal marker,
or two literal markers of different clauses.
Suppose that a strip $S$ includes two
literal markers $\mu$ and $\nu$ of the same clause $C_j$
but of different indices $j,s$ and $j,t$.
Assume without loss of generality that $\mu$ and $\nu$ are consecutive in $S$.
Recall the orders of the literal markers of each clause
in the two genomic maps:

$$
\begin{array}{c}

\zl{j,1} \quad \zr{j,1}
\quad\cdots\quad
\underline{\zl{j,s} \quad \zr{j,s} \quad\cdots\quad \zl{j,t} \quad \zr{j,t}}
\quad\cdots\quad
\zl{j,q} \quad \zr{j,q} \\

\zl{j,q}
\quad\cdots\quad
\underline{\zl{j,t} \quad\cdots\quad \zl{j,s}}
\quad\cdots\quad
\zl{j,1} \quad \zr{j,q}
\quad\cdots\quad
\underline{\zr{j,t} \quad\cdots\quad \zr{j,s}}
\quad\cdots\quad
\zr{j,1}

\end{array}
$$

Since in $G_1$ the pairs of literal markers appear with ascending indices,
the index $s$ of the marker $\mu$ must be less than
the index $t$ of the marker $\nu$.
Then, since in $G_2$ the left markers appear with descending indices
before the right markers also with descending indices,
$\mu$ must be a left marker, and $\nu$ must be a right marker.
That is, $\mu\nu = \;\zl{j,s}\;\zr{j,t}$.
All markers between $\mu$ and $\nu$ in $G_1$
must be missing from the subsequences.
Among these missing markers,
those that are literal markers of $C_j$
appear in $G_2$ either consecutively before $\mu$ or consecutively after $\nu$.
Replace either $\mu$ or $\nu$ by a missing literal marker of $C_j$,
that is, either $\zl{j,s}$ by $\zl{j,t}$, or $\zr{j,t}$ by $\zr{j,s}$,
then $\mu$ and $\nu$ become a pair.
Denote this \emph{shift} operation by
$$
\mu\nu:\quad
\zl{j,s}\;\zr{j,t} \;\to\; \zl{j,t}\;\zr{j,t} \textup{ or } \zl{j,s}\;\zr{j,s}.
$$
The strip $S$ cannot include any other literal markers of the clause $C_j$
besides $\mu$ and $\nu$ because
(i)
the markers before $\zl{j,s}$ in $G_1$ appear after $\zl{j,s}$ in $G_2$,
and (ii)
the markers after $\zr{j,t}$ in $G_1$ appear before $\zr{j,t}$ in $G_2$.

\textbf{\em 6. At most one pair of literal markers of each clause is a strip.}
Note that the $q$ pairs of literal markers of each clause appear in $G_2$
in an interleaving pattern.
It follows by condition~5 that at most one of the $q$ pairs can be a strip.

\textbf{\em 7. Each strip that includes a true/false marker is a pair of true markers or a pair of false markers.}
By conditions 1, 3, and 5,
it follows that each strip that includes a true/false marker
must include true/false markers only.
A strip cannot include two true/false markers of different variables
because they appear in different orders in $G_1$ and in $G_2$.
Suppose that a strip $S$ includes two true/false markers $\mu$ and $\nu$
of the same variable $x_i$ such that $\mu$ and $\nu$ are not a pair.
Recall the orders of the four pairs of true/false markers of each variable $x_i$
in $G_1$ and $G_2$,
the four possible positions
of the three pairs of literal markers in $G_1$,
and the position of the variable marker in $G_2$:

$$
\begin{array}{c}

\zl{}\fl{1}\zr{} \quad \tl{2} \quad \fr{1} \quad \zl{}\tr{2}\zr{}
\quad\quad
\zl{}\fl{2}\zr{} \quad \tl{1} \quad \fr{2} \quad \zl{}\tr{1}\zr{} \\

\tl{1} \quad \fl{1} \quad \tr{1} \quad \fr{1}
\quad \vl{} \quad \vr{} \quad
\fl{2} \quad \tl{2} \quad \fr{2} \quad \tr{2}

\end{array}
$$

Note that the pair of variable markers in $G_2$ forbids a strip
from including two true/false markers of different indices.
Thus the strip $S$ must consist of true/false markers of both the same variable
and the same index.
Assume without loss of generality that $\mu$ appears before $\nu$ in $S$.
It is easy to check that there are only two such combinations
of $\mu$ and $\nu$:
either $\mu\nu = \;\fl{1}\;\tr{1}$ or $\mu\nu = \;\tl{2}\;\fr{2}$.
Moreover, the strip $S$ must include only the two markers $\mu$ and $\nu$.
For either combination of $\mu$ and $\nu$,
use a shift operation to make $\mu$ and $\nu$ a pair:
$$
\begin{array}{ll}

\mu\nu:\quad&
\fl{1}\;\tr{1} \;\to\; \fl{1}\;\fr{1} \textup{ or } \tl{1}\;\tr{1} \\

\mu\nu:\quad&
\tl{2}\;\fr{2} \;\to\; \tl{2}\;\tr{2} \textup{ or } \fl{2}\;\fr{2}.

\end{array}
$$

\textbf{\em 8. Either both pairs of true markers or both pairs of false markers of each variable are two strips.}
Consider the conflict graph of
the four pairs of true/false markers and the three pairs of literal markers
of each variable $x_i$ in Figure~\ref{fig:7}.
The graph has one vertex for each pair,
and has an edge between two vertices if and only if
the corresponding pairs intersect in either $G_1$ or $G_2$.
By conditions 1, 3, 5, and 7,
the strips of the subsequences from the seven pairs
correspond to an independent set in the conflict graph of seven vertices.

\begin{figure}[htb]
\psfrag{S}{$S$}
\psfrag{T}{$T$}
\psfrag{>}{$\Rightarrow$}
\centering
\includegraphics[scale=0.6]{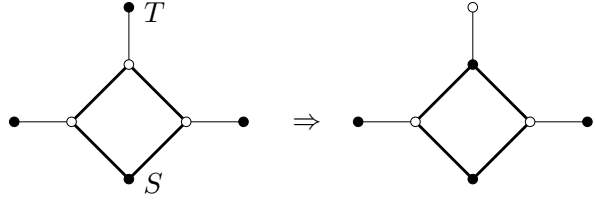}
\caption{\small%
Replacing vertices of the independent set in the conflict graph
of the four pairs of true/false markers
and the three pairs of literal markers of each variable.
Vertices in the independent set are black.
Edges in the $4$-cycle are thick.
In this example the strip $S$ is first deleted then inserted back.}
\label{fig:7}
\end{figure}

Note that the four vertices
corresponding to the four pairs of true/false markers
induce a $4$-cycle in the conflict graph.
Suppose that neither both pairs of true markers
nor both pairs of false markers are strips.
Then at most one of the four pairs, say $S$, is a strip.
Delete $S$ from the subsequences.
Recall that each variable has either $2$ positive and $1$ negative literals,
or $1$ positive and $2$ negative literals.
Let $T$ be the pair of literal markers whose sign is opposite to the sign
of the other two pairs of literal markers.
Also delete $T$ from the subsequences if it is there.
Next insert two pairs of true/false markers to the subsequences:
if $T$ is positive,
both pairs of false markers
$\fl{i,1}\;\fr{i,1}$ and $\fl{i,2}\;\fr{i,2}$;
if $T$ is negative,
both pairs of true markers
$\tl{i,1}\;\tr{i,1}$ and $\tl{i,2}\;\tr{i,2}$.

When all eight conditions are satisfied,
the subsequences are in the desired canonical form.
\end{proof}

The following lemma, analogous to
Lemma~\ref{lem:iff4} and Lemma~\ref{lem:iff3},
establishes the NP-hardness of \msr{2}:

\begin{lemma}\label{lem:iff2}
The variables in $X$
have an assignment that satisfies at least $k$ clauses in $\C$
if and only if
the two genomic maps $G_1$ and $G_2$ have two subsequences
whose total strip length $l$ is at least $2(3n + m + k + 2)$.
\end{lemma}

\begin{proof}
We first prove the ``only if'' direction.
Suppose that
the variables in $X$
have an assignment that satisfies at least $k$ clauses in $\C$.
We will show that
the two genomic maps $G_1$ and $G_2$
have two subsequences of total strip length at least $2(3n + m + k + 2)$.
For each variable $x_i$,
choose the two pairs of true markers if the variable is assigned true,
or the two pairs of false markers if the variable is assigned false.
For each satisfied clause $C_j$,
choose one pair of literal markers corresponding to a true literal
(when there are two or more true literals, choose any one).
Also choose all $m+n$ pairs of clause and variable markers
and both pairs of dummy markers.
The chosen markers induce two subsequences of the two genomic maps.
It is easy to check that, by construction,
the two subsequences have at least $3n + m + k + 2$ strips,
each strip forming a pair.
Thus the total strip length is at least $2(3n + m + k + 2)$.
We refer to Figure~\ref{fig:msr2} (c) and (d) for two examples.

We next prove the ``if'' direction.
Suppose that
the two genomic maps $G_1$ and $G_2$
have two subsequences of total strip length at least $2(3n + m + k + 2)$.
We will show that
the variables in $X$
have an assignment that satisfies at least $k$ clauses in $\C$.
By Lemma~\ref{lem:canon2},
the two genomic maps have two subsequences
of total strip length at least $2(3n + m + k + 2)$
such that each strip is a pair and, moreover,
the two pairs of dummy markers,
the $m+n$ pairs of clause and variable markers,
at most one pair of literal markers of each clause,
and either both pairs of true markers or both pairs of false markers
of each variable are strips.
Thus at least $k$ strips are pairs of literal markers,
each pair of a different clause.
Again it is easy to check that, by construction,
the assignment of the variables in $X$ to either true or false
(corresponding to the choices of
either both pairs of true markers or both pairs of false markers)
satisfies at least $k$ clauses in $\C$
(corresponding to the at least $k$ pairs of literal markers that are strips).
\end{proof}

\subsection{L-reduction from \sat{p}{q} to \msr{2}}

We present an L-reduction $(f, g, \alpha, \beta)$ from \sat{p}{q} to \msr{3}
as follows.
The function $f$,
given the \sat{p}{q} instance $(X,\C)$,
constructs the two genomic maps $G_1$ and $G_2$
as in the NP-hardness reduction.
Let $k^*$ be the maximum number of clauses in $\C$
that can be satisfied by an assignment of $X$,
and let $l^*$ be the maximum total strip length of any two subsequences of
$G_1$ and $G_2$, respectively.
Since a random assignment of each variable independently
to either true or false with equal probability $\frac12$
satisfies each disjunctive clause of $q$ literals with probability
$1-\frac1{2^q}$,
we have $k^* \ge \frac{2^q - 1}{2^q}m$.
By Lemma~\ref{lem:iff2},
we have $l^* = 2(3n + m + k^* + 2)$.
Recall that $np = mq$.
It follows that
$$
l^* = 2(3n + m + k^* + 2)
= \left( 6\,\frac{q}{p} + 2 \right) m + 2k^* + 4
\le \left(
	\left( 6\,\frac{q}{p} + 2 \right) \frac{2^q}{2^q-1} + 2 + \frac4{k^*}
\right) k^*.
$$

The function $g$,
given two subsequences of the two genomic maps
$G_1$ and $G_2$, respectively,
transforms the subsequences into canonical form
as in the proof of Lemma~\ref{lem:canon2},
then returns an assignment of $X$
corresponding to the choices of true or false markers.
Let $l$ be the total strip length of the subsequences,
and let $k$ be the number of clauses in $\C$ that are satisfied
by this assignment.
Then 
$k \ge l/2 - 3n - m - 2$.
It follows that
$$
|k^* - k| = k^* - k \le (l^*/2 - 3n - m - 2) - (l/2 - 3n - m - 2) = |l^* - l|/2.
$$

Let $\epsilon > 0$ be an arbitrary small constant.
Note that by brute force we can check
whether $k^* < 2/\epsilon$ and,
in the affirmative case,
compute an optimal assignment of $X$
that satisfies the maximum number of clauses in $\C$,
all in $m^{O(1/\epsilon)}$ time,
which is polynomial in $m$ for a constant $\epsilon$.
Therefore we can assume without loss of generality that $k^* \ge 2/\epsilon$.
Then,
with the two constants
$\alpha = ( 6\,\frac{q}{p} + 2 ) \frac{2^q}{2^q-1} + 2 + 2\epsilon$
and $\beta = 1/2$,
both properties \eqref{eq:L1} and \eqref{eq:L2} of L-reduction
are satisfied.
In particular,
for $p = 3$ and $q = 2$,
$$
\alpha\beta
= \left( 3\frac{q}{p} + 1 \right) \frac{2^q}{2^q-1} + 1 + \epsilon
= 5 + \epsilon.
$$

Berman and Karpinski~\cite{BK03} showed that
\sat{3}{2} is NP-hard to approximate within any constant less than
$\frac{464}{463} = \frac1{1-1/464}$.
Thus \msr{2} is NP-hard to approximate within any constant less than
$$
\lim_{\epsilon\to0}
\frac1{1-(1/464)/(5+\epsilon)}
= \frac1{1 - 1/2320} = \frac{2320}{2319} = 1.000431\ldots.
\qquad
$$

\section{An asymptotic lower bound for \msr{d}}
\label{sec:msrd}

In this section, we derive an asymptotic lower bound for approximating \msr{d}
by an L-reduction from \dm{d} to \msr{(d+2)}.

\subsection{NP-hardness reduction from \dm{d} to \msr{(d+2)}}

Let $E \subseteq V_1 \times \cdots \times V_d$ be a set of $n$ hyper-edges
over $d$ disjoint sets $V_i$ of vertices, $1 \le i \le d$.
We construct two genomic maps $G_{\rightarrow}$ and $G_{\leftarrow}$,
and $d$ genomic maps $G_i$, $1 \le i \le d$,
where each map is a permutation of the following $2n$ distinct markers
all in positive orientation:
\begin{itemize}

\item
$n$ pairs of edge markers $\zl i$ and $\zr i$, $1 \le i \le n$.

\end{itemize}

The two genomic maps $G_{\rightarrow}$ and $G_{\leftarrow}$
are concatenations of the $n$ pairs of edge markers
with ascending and descending indices, respectively:

$$
\begin{array}{lc}

G_{\rightarrow}: & \zl 1 \; \zr 1 \quad \cdots \quad \zl n \; \zr n \\

G_{\leftarrow}: & \zl n \; \zr n \quad \cdots \quad \zl 1 \; \zr 1

\end{array}
$$

Each genomic map $G_i$ corresponds to a vertex set
$V_i = \{ v_{i,j} \mid 1 \le j \le |V_i| \}$,
$1 \le i \le d$,
and is represented schematically as follows:
$$
G_i: \quad \cdots \quad \langle v_{i,j} \rangle \quad \cdots \quad
$$
Here each $\langle v_{i,j} \rangle$
consists of the edge markers of hyper-edges containing the vertex $v_{i,j}$,
grouped together such that
the left markers appear with ascending indices
before the right markers also with ascending indices.
This completes the construction.
We refer to Figure~\ref{fig:msrd}(a) for an example.

\begin{figure}[htb]
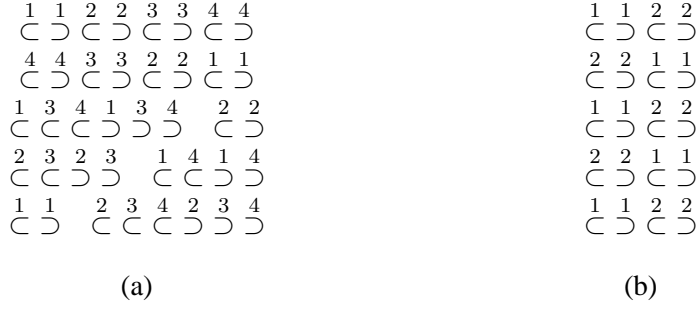

\centering

\begin{minipage}[c]{0.4\linewidth}
\centering
$$
\begin{array}{c}

\zl 1 \; \zr 1 \; \zl 2 \; \zr 2 \; \zl 3 \; \zr 3 \; \zl 4 \; \zr 4 \\

\zl 4 \; \zr 4 \; \zl 3 \; \zr 3 \; \zl 2 \; \zr 2 \; \zl 1 \; \zr 1 \\

\zl 1 \; \zl 3 \; \zl 4 \; \zr 1 \; \zr 3 \; \zr 4 \quad \zl 2 \; \zr 2 \\

\zl 2 \; \zl 3 \; \zr 2 \; \zr 3 \quad \zl 1 \; \zl 4 \; \zr 1 \; \zr 4 \\

\zl 1 \; \zr 1 \quad \zl 2 \; \zl 3 \; \zl 4 \; \zr 2 \; \zr 3 \; \zr 4

\end{array}
$$
(a)
\end{minipage}
\begin{minipage}[c]{0.4\linewidth}
\centering
$$
\begin{array}{c}

\zl 1 \; \zr 1 \; \zl 2 \; \zr 2 \\

\zl 2 \; \zr 2 \; \zl 1 \; \zr 1 \\

\zl 1 \; \zr 1 \; \zl 2 \; \zr 2 \\

\zl 2 \; \zr 2 \; \zl 1 \; \zr 1 \\

\zl 1 \; \zr 1 \; \zl 2 \; \zr 2

\end{array}
$$
(b)
\end{minipage}

\caption{\small%
\msr{5} construction for the \dm{3} instance
$V_1 = \{ v_{1,1}, v_{1,2} \}$,
$V_2 = \{ v_{2,1}, v_{2,2} \}$,
$V_3 = \{ v_{3,1}, v_{3,2} \}$,
and
$E = \{\,
e_1 = (v_{1,1}, v_{2,2}, v_{3,1}),\,
e_2 = (v_{1,2}, v_{2,1}, v_{3,2}),\,
e_3 = (v_{1,1}, v_{2,1}, v_{3,2}),\,
e_4 = (v_{1,1}, v_{2,2}, v_{3,2})
\,\}$.
(a) The five genomic maps $G_{\rightarrow}, G_{\leftarrow}, G_1, G_2, G_3$.
(b) The five subsequences of the genomic maps
corresponding to the subset $\{ e_1, e_2 \}$ of pairwise-disjoint hyper-edges.}
\label{fig:msrd}
\end{figure}

The following property of our construction is obvious:

\begin{proposition}\label{prp:msrd}
Two hyper-edges in $E$ intersect if and only if
the corresponding two pairs of edge markers intersect
in one of the $d$ genomic maps $G_i$, $1 \le i \le d$.
\end{proposition}

The following lemma is analogous to Lemma~\ref{lem:canon4}:

\begin{lemma}\label{lem:canond}
In any $d+2$ subsequences of the $d+2$ genomic maps
$G_{\rightarrow}, G_{\leftarrow}, G_1, \ldots, G_d$, respectively,
each strip must be a pair of edge markers.
\end{lemma}

\begin{proof}
By construction,
a strip cannot include two edge markers of different indices
because they appear in different orders
in $G_{\rightarrow}$ and in $G_{\leftarrow}$.
\end{proof}

The following lemma, analogous to
Lemma~\ref{lem:iff4}, Lemma~\ref{lem:iff3}, and Lemma~\ref{lem:iff2},
establishes the NP-hardness of \msr{d}:

\begin{lemma}\label{lem:iffd}
The set $E$ has a subset of $k$ pairwise-disjoint hyper-edges
if and only if
the $d+2$ genomic maps
$G_{\rightarrow}, G_{\leftarrow}, G_1, \ldots, G_d$
have $d+2$ subsequences
whose total strip length $l$ is at least $2k$.
\end{lemma}

\begin{proof}
We first prove the ``only if'' direction.
Suppose that
the set $E$ has a subset of at least $k$ pairwise-disjoint hyper-edges.
We will show that
the $d+2$ genomic maps
$G_{\rightarrow}, G_{\leftarrow}, G_1, \ldots, G_d$
have $d+2$ subsequences of total strip length at least $2k$.
By Proposition~\ref{prp:msrd},
the $k$ pairwise-disjoint hyper-edges
correspond to $k$ pairs of edge markers that do not intersect
each other in the genomic maps.
These $k$ pairs of edge markers
induce a subsequence of length $2k$ in each genomic map.
In each subsequence,
the left marker and the right marker of each pair appear consecutively
and compose a strip.
Thus the total strip length is at least $2k$.
We refer to Figure~\ref{fig:msrd}(b) for an example.

We next prove the ``if'' direction.
Suppose that
the $d+2$ genomic maps
$G_{\rightarrow}, G_{\leftarrow}, G_1, \ldots, G_d$
have $d+2$ subsequences of total strip length at least $2k$.
We will show that
the set $E$ has a subset of at least $k$ pairwise-disjoint hyper-edges.
By Lemma~\ref{lem:canon4},
each strip of the subsequences must be a pair of edge markers.
Thus we obtain at least $k$ pairs of edge markers
that do not intersect each other in the genomic maps.
Then, by Proposition~\ref{prp:msrd},
the corresponding set of at least $k$ hyper-edges in $E$
are pairwise-disjoint.
\end{proof}

\subsection{L-reduction from \dm{d} to \msr{(d+2)}}

We present an L-reduction $(f, g, \alpha, \beta)$ from \dm{d} to \msr{(d+2)}
as follows.
The function $f$,
given a set $E \subseteq V_1 \times \cdots \times V_d$ of hyper-edges,
constructs the $d+2$ genomic maps
$G_{\rightarrow}, G_{\leftarrow}, G_1, \ldots, G_d$
as in the NP-hardness reduction.
Let $k^*$ be the maximum number of pairwise-disjoint hyper-edges in $E$,
and let $l^*$ be the maximum total strip length of any $d+2$ subsequences of
$G_{\rightarrow}, G_{\leftarrow}, G_1, \ldots, G_d$, respectively.
By Lemma~\ref{lem:iffd},
we have
$$
l^*= 2k^*.
$$
Choose $\alpha = 2$,
then property~\eqref{eq:L1} of L-reduction is satisfied.

The function $g$,
given $d+2$ subsequences of the $d+2$ genomic maps
$G_{\rightarrow}, G_{\leftarrow}, G_1, \ldots, G_d$, respectively,
returns a subset of pairwise-disjoint hyper-edges in $E$
corresponding to
the pairs of edge markers that are strips of the subsequences.
Let $l$ be the total strip length of the subsequences,
and let $k$ be the number of pairwise-disjoint hyper-edges
returned by the function $g$.
Then $k \ge l/2$.
It follows that
$$
|k^* - k| = k^* - k \le l^*/2 - l/2 = |l^* - l|/2.
$$
Choose $\beta = 1/2$,
then property~\eqref{eq:L2} of L-reduction is also satisfied.

We have obtained an L-reduction from \dm{d} to \msr{(d+2)}
with $\alpha\beta = 1$.
Hazan, Safra, and Schwartz~\cite{HSS06} showed that
\dm{d} is NP-hard to approximate within $\Omega(d/\log d)$.
It follows that \msr{d} is also NP-hard to approximate within
$\Omega(d/\log d)$.
This completes the proof of Theorem~\ref{thm:msrd}.

\section{A polynomial-time $2d$-approximation for \msr{d}}
\label{sec:2d}

In this section we prove Theorem~\ref{thm:2d}.
We briefly review the two previous algorithms~\cite{ZZS07,CFJZ09}
for this problem.
The first algorithm for \msr{2} is a simple heuristic due to
Zheng, Zhu, and Sankoff~\cite{ZZS07}:
\begin{enumerate}
\item Extract a set of pre-strips from the two genomic maps;
\item Compute an independent set of strips from the pre-strips.
\end{enumerate}
This algorithm is inefficient because the number of pre-strips could be
exponential in the sequence length, and furthermore the problem
Maximum-Weight Independent Set in general graphs is NP-hard.

Chen, Fu, Jiang, and Zhu~\cite{CFJZ09} presented
a $2d$-approximation algorithm for \msr{d}.
For any $d \ge 2$,
a \emph{$d$-interval} is the union of $d$ disjoint intervals in the real line,
and a \emph{$d$-interval graph} is the intersection graph of
a set of $d$-intervals, with a vertex for each $d$-interval,
and with an edge between two vertices
if and only the corresponding $d$-intervals overlap.
The $2d$-approximation algorithm~\cite{CFJZ09} works as follows:
\begin{enumerate}
\item Compose a set of $d$-intervals, one for each combination
of $d$ substrings of the $d$ genomic maps, respectively.
Assign each $d$-interval a weight equal to
the length of a longest common subsequence (which may be reversed and negated)
in the corresponding $d$ substrings.
\item Compute a $2d$-approximation for Maximum-Weight Independent Set
in the resulting $d$-interval graph
using Bar-Yehuda \etal's fractional local-ratio algorithm~\cite{BHN+06}.
\end{enumerate}

Let $n$ be the number of markers in each genomic map.
Then the number of $d$-intervals composed by this algorithm
is $\Theta(n^{2d})$
because each of the $d$ genomic maps has $\Theta(n^2)$ substrings.
Consequently the running time of this algorithm can be exponential
if the number $d$ of genomic maps is not a constant but is part of the input.
In the following, we show that if all markers are distinct in each genomic map
(as discussed earlier, this is a reasonable assumption in application),
then the running time of the $2d$-approximation algorithm can be improved
to polynomial for all $d \ge 2$.
This improvement is achieved by composing a smaller set of candidate
$d$-intervals in step~1 of the algorithm.

The idea is actually quite simple
and has been used many times previously~\cite{Ly04,Ji10b,BFR09}.
Note that any strip of length $l > 3$ is a concatenation of shorter
strips of lengths $2$ and $3$, for example, $4 = 2 + 2$, $5 = 2 + 3$, etc.
Since the objective is to maximize the total strip length,
it suffices to consider only short strips of lengths $2$ and $3$
in the genomic maps,
and to enumerate only candidate $d$-intervals that correspond to these strips.
When each genomic map is a signed permutation of the same $n$ distinct markers,
there are at most ${n \choose 2} + {n \choose 3} = O(n^3)$ strips
of lengths $2$ and $3$,
and for each strip there is a unique shortest substring of each genomic map
that contains all markers in the strip.
Thus we compose only $O(n^3)$ $d$-intervals,
and improve the running time of the $2d$-approximation algorithm to polynomial
for all $d \ge 2$.
This completes the proof of Theorem~\ref{thm:2d}.

\section{Inapproximability results for related problems}
\label{sec:more}

In this section we prove Theorem~\ref{thm:cmsrd} and Theorem~\ref{thm:gap}.

\paragraph{\cmsr{3} and \cmsr{4} are APX-hard.}

For any $d$, the decision problems of \msr{d} and \cmsr{d} are equivalent.
Thus the NP-hardness of $\msr{d}$ implies the NP-hardness of $\cmsr{d}$,
although the APX-hardness of $\msr{d}$ does not necessarily
imply the APX-hardness of $\cmsr{d}$.
Note that the two problems \mis{\Delta} and \mvc{\Delta} complement each other
just as the two problems \msr{d} and \cmsr{d} complement each other.
Thus our NP-hardness reduction from \mis{3} to \msr{3}
in Section~\ref{sec:msr3} can be immediately turned into
an NP-hardness reduction from \mvc{3} to \cmsr{3}.

We present an L-reduction $(f, g, \alpha, \beta)$ from \mvc 3 to \cmsr 3
as follows.
The function $f$,
given a graph $G$ of maximum degree $3$,
constructs the three genomic maps $G_0, G_1, G_2$
as in the NP-hardness reduction in Section~\ref{sec:msr3}.
Let $k^*$ be the number of vertices in a maximum independent set in $G$,
and let $l^*$ be the maximum total strip length of any three subsequences of
$G_0, G_1, G_2$, respectively.
Also let $c^*$ be the number of vertices in a minimum vertex cover in $G$,
and let $x^*$ be the minimum number of markers that must be deleted
to transform the three genomic maps $G_0, G_1, G_2$
into strip-concatenated subsequences.
Then $k^* + c^* = n$ and $l^* + x^* = 4n$.
By Lemma~\ref{lem:iff3},
we have
$l^* = 2(n + k^*)$.
It follows that
$$
x^*= 4n - l^* = 4n - 2(n + k^*) = 2(n - k^*) = 2c^*.
$$
Choose $\alpha = 2$,
then property~\eqref{eq:L1} of L-reduction is satisfied.

The function $g$,
given three subsequences of the three genomic maps $G_0, G_1, G_2$,
respectively,
transforms the subsequences into canonical form
as in the proof of Lemma~\ref{lem:canon3},
then returns a vertex cover in the graph $G$
corresponding to the deleted pairs of vertex markers.
Let $x$ be the number of deleted vertex markers,
and let $c$ be the number of vertices in the vertex cover
returned by the function $g$.
Then $c \le x/2$.
It follows that
$$
|c^* - c| = c - c^* \le x/2 - x^*/2 = |x^* - x|/2.
$$
Choose $\beta = 1/2$,
then property~\eqref{eq:L2} of L-reduction is also satisfied.

The L-reduction from \mvc{3} to \cmsr{3} can be obviously generalized:

\begin{lemma}\label{lem:mvc}
Let $\Delta \ge 3$ and $d \ge 3$.
If there is a polynomial-time algorithm for
decomposing any graph of maximum degree $\Delta$ into $d-1$ linear forests,
then there is an L-reduction from \mvc{\Delta} to \cmsr{d}
with constants $\alpha=2$ and $\beta=1/2$.
\end{lemma}

Recall that there exist polynomial-time algorithms for
decomposing a graph of maximum degree $3$ and $4$
into at most $2$ and $3$ linear forests,
respectively~\cite{AEH80,AC81,AEH81}.
Thus we have an L-reduction from \mvc{3} to \cmsr{3} and
an L-reduction from \mvc{4} to \cmsr{4},
with the same parameters $\alpha=2$, $\beta=1/2$, and $\alpha\beta = 1$.
Chleb\'ik and Chleb\'ikov\'a~\cite{CC06} showed that
\mvc{3} and \mvc{4}
are NP-hard to approximate within
$1.0101215$
and
$1.0202429$,
respectively.
It follows that
\cmsr{3} and \cmsr{4}
are NP-hard to approximate within
$1.0101215$
and
$1.0202429$,
respectively, too.
The lower bound for \cmsr{4} extends to \cmsr{d} for all $d \ge 4$.
Note that we could use an L-reduction from \mvc{3} to \cmsr{4}
similar to the L-reduction from \mis{3} to \msr{4} in Section~\ref{sec:msr4},
but that only gives us a weaker lower bound of $1.0101215$ for \cmsr{4}.

\paragraph{\cmsr{2} is APX-hard.}

Let $p = 3$ and $q \ge 2$.
We present an L-reduction $(f, g, \alpha, \beta)$ from \sat{p}{q} to \cmsr{2}
as follows.
The function $f$,
given the \sat{p}{q} instance $(X,\C)$,
constructs the two genomic maps $G_1$ and $G_2$
as in our NP-hardness reduction in Section~\ref{sec:msr2}.
As before,
let $k^*$ be the maximum number of clauses in $\C$
that can be satisfied by an assignment of $X$,
and let $l^*$ be the maximum total strip length of any two subsequences of
$G_1$ and $G_2$, respectively.
Also let $x^*$ be the minimum number of deleted markers.
Then $l^* + x^*$ is exactly the number of markers in each genomic map,
that is, $2(5n + m + qm + 2)$.
By Lemma~\ref{lem:iff2},
we have $l^* = 2(3n + m + k^* + 2)$.
Thus $x^* = 2(5n + m + qm + 2) - 2(3n + m + k^* + 2) = 2(2n + qm - k^*)$.
Since a random assignment of each variable independently
to either true or false with equal probability $\frac12$
satisfies each disjunctive clause of $q$ literals with probability
$1-\frac1{2^q}$,
we have $k^* \ge \frac{2^q - 1}{2^q}m$.
Recall that $np = mq$.
It follows that
$$
x^* = 2(2n + qm - k^*)
= 2\left( 2\,\frac{q}{p} + q \right) m - 2k^*
\le \left(
	2\left( 2\,\frac{q}{p} + q \right) \frac{2^q}{2^q-1} - 2
\right) k^*.
$$
For $p = 3$ and $q = 2$, we can choose
$\alpha = 2( 2\,\frac{q}{p} + q ) \frac{2^q}{2^q-1} - 2 = 62/9$.
Then property~\eqref{eq:L1} of L-reduction is satisfied.

The function $g$,
given two subsequences of the two genomic maps $G_1$ and $G_2$,
transforms the subsequences into canonical form
as in the proof of Lemma~\ref{lem:canon2},
then returns an assignment of $X$
corresponding to the choices of true or false markers.
Let $l$ be the total strip length of the subsequences,
and let $x$ be the number of deleted markers.
Let $k$ be the number of clauses in $\C$ that are satisfied
by this assignment.
Then
$$
|k^* - k| \le |l^* - l|/2 = |x^* - x|/2.
$$
Choose $\beta = 1/2$.
then property~\eqref{eq:L2} of L-reduction is satisfied.

Berman and Karpinski~\cite{BK03} showed that
\sat{3}{2} is NP-hard to approximate within any constant less than
$\frac{464}{463} = \frac1{1-1/464}$.
Since $\alpha\beta = 31/9$,
\cmsr{2} is NP-hard to approximate within any constant less than
$$
1+(1/464)/(31/9) = 1 + 9/14384 = 1.000625\ldots.
\qquad
$$

\paragraph{An asymptotic lower bound for \cmsr{d} and a lower bound for \cmsr{d} with unbounded $d$.}

Chleb\'ik and Chleb\'ikov\'a~\cite{CC06} showed that for any $\Delta \ge 228$,
\mvc{\Delta} is NP-hard to approximate within
$\frac{7}{6} - O(\log \Delta /\Delta)$.
By the second inequality in~\eqref{eq:f},
it follows that if $\Delta \le 227$,
then $f(\Delta) \le \lceil 3\lceil 227/2 \rceil/2 \rceil = 171$.
Consequently, if $f(\Delta) \ge 172$, then $\Delta \ge 228$.
By Lemma~\ref{lem:mvc}, there is an L-reduction
from \mvc{\Delta} to \cmsr{(f(\Delta)+1)} with $\alpha = 2$ and $\beta = 1/2$.
Therefore,
for any $d \ge 173$,
\cmsr{d} is NP-hard to approximate within $\frac{7}{6}- O(\log d /d)$.

The maximum degree $\Delta$ of a graph of $n$ vertices is at most $n-1$.
Again by the second inequality in~\eqref{eq:f},
we have
$f(\Delta) \le \lceil 3\lceil (n-1)/2 \rceil/2 \rceil$.
Thus $f(\Delta)$ is bounded by a polynomial in $n$.
If $d$ is not a constant but is part of the input,
then a straightforward generalization of
the L-reduction from \mvc{3} to \cmsr{3} as in Lemma~\ref{lem:mvc}
gives an L-reduction from Minimum Vertex Cover to \cmsr{(f(\Delta)+1)}
with $\alpha = 2$ and $\beta = 1/2$.
Dinur and Safra~\cite{DS05} showed that Minimum Vertex Cover
is NP-hard to approximate within any constant less than
$10\sqrt5 - 21 = 1.3606\ldots$.
It follows that if $d$ is not a constant but is part of the input,
then \cmsr{d} is NP-hard to approximate within any constant less than
$10\sqrt5 - 21 = 1.3606\ldots$.
This completes the proof of Theorem~\ref{thm:cmsrd}.

\paragraph{Inapproximability of \gapmsr{\delta}{d} and \gapcmsr{\delta}{d}.}

It is easy to check that
all instances of \msr{d} and \cmsr{d} in our constructions for
Theorem~\ref{thm:msrd} and Theorem~\ref{thm:cmsrd}
admit optimal solutions in canonical form with maximum gap $2$,
except for the following two cases:
\begin{enumerate}

\item
In the L-reduction from \sat{p}{q} to \msr{2} and \cmsr{2},
a strip that is a pair of literal markers has a gap of $q - 1$,
which is larger than $2$ for $q \ge 4$.

\item
In the L-reduction from \dm{d} to \msr{(d+2)},
a strip that is a pair of edge markers may have an arbitrarily large gap
if it corresponds to one of many hyper-edges that share a single vertex.

\end{enumerate}

To extend our results in Theorem~\ref{thm:msrd} and Theorem~\ref{thm:cmsrd}
to the corresponding results in Theorem~\ref{thm:gap},
the first case does not matter because
we set the parameter $q$ to $2$
when deriving the lower bounds for \msr{2} and \cmsr{2}
from the lower bound for \sat{3}{2}.

The second case is more problematic, and we have to use a different L-reduction
to obtain a slightly weaker asymptotic lower bound for \gapmsr{\delta}{d}.
Trevisan~\cite{Tr01} showed that
\mis{\Delta} is NP-hard to approximate within
$\Delta/2^{O(\sqrt{\log \Delta})}$.
By Lemma~\ref{lem:mis},
there is an L-reduction from \mis{\Delta} to \gapmsr{\delta}{(f(\Delta)+2)}
with $\alpha\beta = 1$.
By the two inequalities in~\eqref{eq:f},
we have $f(\Delta) + 2 = \Theta(\Delta)$.
Thus \gapmsr{\delta}{d} is NP-hard to approximate within
$d/2^{O(\sqrt{\log d})}$.
This completes the proof of Theorem~\ref{thm:gap}.

\section{Concluding remarks}
\label{sec:remarks}

A strip of length $l$ has $l - 1$ \emph{adjacencies}
between consecutive markers.
In general, $k$ strips of total length $l$ have $l-k$ adjacencies.
Besides the total strip length,
the total number of adjacencies in the strips is also
a natural objective function of \msr{d}~\cite{CFJZ09}.
It can be checked that our L-reductions for \msr{d} and \gapmsr{\delta}{d}
still work
even if the objective function is changed from the total strip length
to the total number of adjacencies in the strips.
The only effect of this change is that the constant $\alpha$ is halved
and correspondingly the constant $\beta$ is doubled (from $1/2$ to $1$).
Since the product $\alpha\beta$ is unaffected,
Theorem~\ref{thm:msrd} and the second part of Theorem~\ref{thm:gap}
remain valid.
For Theorem~\ref{thm:2d}, we can adapt the $2d$-approximation algorithm
for maximizing the total strip length to
a $(2d+\epsilon)$-approximation algorithm for
maximizing the total number of adjacencies in strips,
for any constant $\epsilon > 0$.
The only change in the algorithm is to enumerate all $d$-intervals
of strip lengths at most $\Theta(1/\epsilon)$,
instead of $2$ and $3$.
We note that the small difference between the two objective functions,
total length versus total number of adjacencies,
has led to difference in the complexities
of two other bioinformatics problems~\cite{Ly04,Ji10b}:
For RNA secondary structure prediction,
the problem
Maximum Stacking Base Pairs (MSBP)
maximizes the total length of helices,
and the problem
Maximum Base Pair Stackings (MBPS)
maximizes the total number of adjacencies in helices.
On implicit input of base pairs determined by pair types,
MSBP is polynomially solvable,
but
MBPS is NP-hard and admits a polynomial-time approximation scheme~\cite{Ly04};
on explicit input of base pairs,
MSBP and MBPS are both NP-hard,
and admit constant approximations with factors $5/2$ and $8/3$,
respectively~\cite{Ji10b}.

In our Theorem~\ref{thm:msrd} and Theorem~\ref{thm:cmsrd}, we have chosen
to display explicit lower bounds for \msr{2} and \cmsr{2}, despite the fact
that they are rather small and unimpressive.
As commented by M. Karpinski after the author's ISAAC presentation,
it may be possible to improve the lower bound for \msr{2}
by an L-reduction from another problem.
For example, Berman and Karpinski~\cite{BK03} proved that
\sat{3}{2} is APX-hard to approximate within any constant less than
$\frac{464}{463}$ by an L-reduction from
E$d$-Occ-E$k$-LIN-$2$,
and proved that
E$d$-Occ-E$k$-LIN-$2$
is NP-hard to approximate within some other constant by an L-reduction
from yet another problem, and so on.
By constructing an L-reduction directly from
E$d$-Occ-E$k$-LIN-$2$ to \msr{2}, say, we might obtain a better lower bound.
We were not engaged in such pursuits in this paper.
Since satisfiability problems are well-known,
we chose an L-reduction from \sat{3}{2} to \msr{2}
for the sake of a gentle presentation,
and we made no effort in optimizing the constants.

We proved Theorem~\ref{thm:gap} by
extending our proofs of Theorem~\ref{thm:msrd} and Theorem~\ref{thm:cmsrd}
with minimal modifications.
We note that the $\delta$-gap constraint actually makes it easier
to prove the APX-hardness of \gapmsr{\delta}{d} and \gapcmsr{\delta}{d}
than to prove the APX-hardness of \msr{d} and \cmsr{d}.
For example,
our \sat{3}{2} constructions for \msr{2} and \cmsr{2}
can be much simplified to obtain better approximation lower bounds
for \gapmsr{\delta}{d} and \gapcmsr{\delta}{d}.
We omit the details and refer to~\cite{BFR09} for more results
on these restricted variants.
On the other hand, the correctness of our reductions does require
gaps of at least $2$ markers.
Thus our proofs do not imply the APX-hardness of
\gapmsr{1}{d} or \gapcmsr{1}{d}.
Consistent with our results,
Bulteau, Fertin, and Rusu~\cite{BFR09} proved that
\gapmsr{\delta}{2} is APX-hard for all $\delta \ge 2$
and is NP-hard for $\delta = 1$.

A curious concept called \emph{paired approximation}
was recently introduced by Eppstein~\cite{Ep10}.
For certain problems on the same input, say Clique and Independent Set
on the same graph,
sometimes we would be happy to find a good approximation to either one,
if not both.
Inapproximability results for pairs of problems are often \emph{incompatible}:
the hard instances for one problem are disjoint from the hard instances
for the other problem.
As a result,
an approximation algorithm may find a solution to one or the other
of two problems on the same input that is better than the known
inapproximablity bounds for either individual problem.
Note that our inapproximability results for \msr{2} and \cmsr{2} are compatible
because they are obtained from the same reduction from \sat{3}{2}.
Thus even as a paired approximation problem,
(\msr{2}, \cmsr{2}) is still APX-hard.
This is the first inapproximability result
for a paired approximation problem in bioinformatics.

\paragraph{Postscript.}
The APX hardness results for \msr{2} and \msr{3} in Theorem~\ref{thm:msrd}
was obtained in December 2008.
The author was later informed by Binhai Zhu in January 2009
that Lusheng Wang and he
had independently and almost simultaneously proved a weaker result
that \msr{2} is NP-hard~\cite{WZ09}.

\end{document}